\pdfoutput=1
\documentclass[journal]{IEEEtran}

\usepackage[utf8]{inputenc}
\usepackage[english]{babel}
\usepackage[T1]{fontenc}
\usepackage{amsmath,amssymb,amsfonts}
\usepackage{algorithmic}

\newtheorem{lemma}{Lemma}
\newtheorem{theorem}{Theorem}
\newtheorem{remark}{Remark}
\usepackage{graphicx}
\usepackage{braket}
\usepackage{textcomp}
\usepackage{xcolor}
\usepackage{cite}
\usepackage[hidelinks]{hyperref}

\begin{document}

\title{HI-QDC: An Isometric Modular Scaling for Hop-Independent Quantum Data Center}

\author{Mohadeseh~Azari, Anoosha~Fayyaz, Amy~Babay, David~Tipper, Prashant~Krishnamurthy, and Kaushik~Seshadreesan%
\thanks{The authors are with the Department of Informatics and Networked Systems, University of Pittsburgh, Pittsburgh, PA, USA. Corresponding author: Mohadeseh Azari (e-mail: moa125@pitt.edu).}%
}

\markboth{IEEE Journal,~Vol.~XX, No.~XX, 2026}%
{Azari \MakeLowercase{\textit{et al.}}: Hop-Independent Rate and Fidelity}

\maketitle

\begin{abstract}
Server-centric quantum data-center architectures offer scalability by distributing communication tasks across QPUs rather than concentrating complexity in a centralized switching core. However, scaling up such architectures increases the path length, the number of Bell-state measurements, and the loss of end-to-end fidelity. We ask whether the path diversity of a server-centric topology can be converted into a mechanism for preserving not only rate, but also fidelity.
We study this question through end-to-end purification as a fidelity-restoration mechanism. First, in a black-box network model, we determine the minimum number of raw end-to-end Werner-state copies, each distributed over a distance-$\ell$ path and therefore carrying the correspondingly degraded Werner parameter, that purification must consume to recover a single copy whose Werner parameter matches that of an elementary link. Then, comparing the result across recursive $2$-to-$1$ purification and optimized nested $r$-to-$1$ purification. Second, we instantiate these requirements in a probabilistic BCube architecture, since the edge-disjoint path diversity of this topology provides multiple raw copies that can be used for purification. Purification imposes a fundamental lower bound on the fidelity of its input states, and therefore no amount of path redundancy can raise the output fidelity below this threshold. This bound in turn limits how far the network can be scaled up, and to address it we present the Hop-Independent Quantum Data Center (HI-QDC) an isometric, modular, scalable architecture for quantum data centers, where purification transforms end-to-end entanglement distributed across a module into an effective link-level entanglement resource for the inter-module topology. Since BCube emerges as a strong topology candidate, we study the resource region of a HI-QDC-based architecture with BCube as its modular base topology.
Our results identify the regimes in which a BCube module, under end-to-end purification, preserves both the fidelity and the yield — the expected number of entanglement copies — of an elementary link. In this regime, the module hides its internal hop count and can be treated as an effective elementary link for a higher-level network. Thus, topology supplies the path multiplicity required for purification, while purification converts this multiplicity into fidelity recovery, enabling a modular route toward recursively scalable quantum data-center networks.
\end{abstract}

\begin{IEEEkeywords}
Entanglement distribution, quantum data center network, entanglement purification.
\end{IEEEkeywords}

\IEEEpeerreviewmaketitle

\section{Introduction}

Future quantum data center networks are emerging as a specialized sub-network of the global quantum internet~\cite{Kimble2008QuantumInternet,Wehner2018QuantumInternet,Cirac1999Distributed,Cacciapuoti2020QuantumInternetChallenges,Caleffi2024Survey}. Like any entanglement distribution network, they contend with the same fundamental obstacles, but at a different scale, with a different dominant bottleneck, and under a different performance priority, driving recent research to pivot toward data center focused architectures~\cite{Shapourian2025QDC,Cacciapuoti2025QDC}. In such systems, scalability depends not only on whether distant QPUs can be connected, but also on whether the delivered entangled states retain sufficient quality for distributed quantum tasks. This requirement creates an architectural tradeoff. Switch-centric optical fabrics can preserve relatively high end-to-end entanglement fidelity because communication paths require fewer entanglement-swapping operations. However, concentrating communication in a centralized switching core can make scaling difficult, since larger systems require increasingly capable core switches and more complex upgrades. Server-centric architectures take the opposite direction: QPUs participate actively in communication, reducing reliance on a deep centralized core and offering a more distributed path toward scalability~\cite{Guo2008DCell,Guo2009BCube,pouryousef2026benchmarking}. These benefits become especially relevant to future quantum data centers as quantum memories mature and QPUs can serve not only as endpoints but also as intermediate communication resources. However, as the system scales, the increasing reliance on intermediate QPUs can create a fundamental fidelity bottleneck.

Throughout this work, scaling refers to increasing the number of interconnected QPUs. We distinguish between two modes of topological scaling according to their requirements on switch radix and their effect on network diameter. We use scaling out to denote diameter-preserving growth, in which additional QPUs are accommodated within the same topological depth by increasing the switch radix. We use scaling up to denote diameter-growing expansion, in which additional levels, dimensions, or modules are constructed using switches of the same radix. This terminology differs from the conventional system-level distinction in classical data centers~\cite{vahdat2010scale}.

In server-centric architectures, multi-hop connections must traverse many intermediate QPUs and Bell-state measurements, so the fidelity of the distributed end-to-end entanglement degrades as the hop count grows. This issue becomes particularly important when the network must move beyond simple \emph{scaling out}. That is, a flat expansion cannot continue indefinitely without practical limits in the switching infrastructure, routing complexity, or physical organization of the system~\cite{azari2026hopindependentfidelityquantumdata}. At larger scales, some form of \emph{scaling up} or \emph{modular composition} becomes necessary. In server-centric topologies, this increases the network dimension, the diameter, and the number of swapping operations~\cite{Guo2009BCube}. Thus, the key question is whether the fidelity weakness of server-centric architectures can be compensated as they scale.

\begin{figure*}
    \centering
    \includegraphics[width=0.78\linewidth]{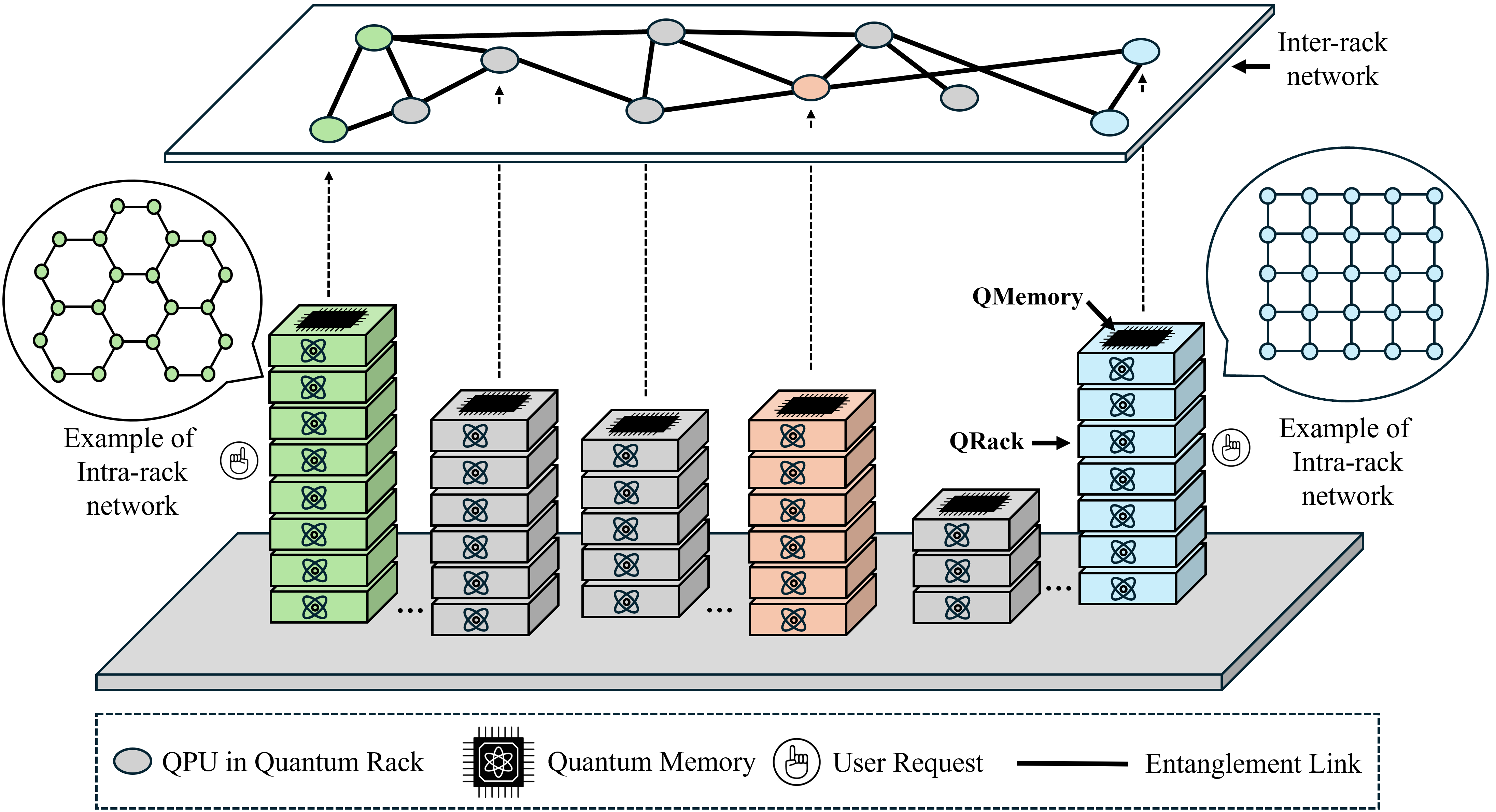}
    \caption{Quantum data-center network architecture considered in this work. The data center is organized into quantum racks, each containing multiple QPUs and quantum memories. Each rack may use a different intra-rack connectivity pattern, chosen according to the number of QPUs, the elementary-link quality, and the operating regime required for local end-to-end entanglement generation. Across racks, we consider a high-node-degree, server-centric inter-rack topology in which QPUs themselves participate in routing and forwarding. Inter-rack connection requests are routed through this topology, whose path diversity provides multiple edge-disjoint routes between the requested endpoints and lets the network generate several raw end-to-end entangled states. These states are stored in quantum memories and consumed by multi-copy purification, which compensates for the fidelity loss incurred during multi-hop entanglement swapping and raises the fidelity of the final shared state.}
    \label{fig:qdc_network}
\end{figure*}

Prior work has shown that network topology can improve the distribution rate of end-to-end entanglement through path diversity, and has characterized the parameter regions that secure such behavior using percolation theory, routing mechanisms, and multiplexing~\cite{Acin2007Percolation,Cuquet2009ComplexPercolation,Perseguers2013LargeScale,Pirandola2019Capacities,Pant2019Routing,Patil2021SpaceTimeMultiplexed,Patil2022DistanceIndependent,Kaur2023SquareGrid}. When a topology provides sufficient path diversity, distant node pairs may access many alternative end-to-end routes, allowing the supported entanglement distribution rate to avoid direct decay with network size. In favorable regimes, the total end-to-end rate can become comparable to, or even exceed, the elementary-link rate. These observations suggest that topology is not merely a connectivity design choice, but a communication resource. Nevertheless, rate resilience alone does not constitute scalability~\cite{Briegel1998Repeaters,Shi2020ConcurrentRouting,Li2022FidelityGuaranteedRouting,Abane2024RoutingSurvey,Kaur2023SquareGrid}. Even if many paths are available, each raw end-to-end state still suffers from fidelity loss accumulated through multi-hop swapping.

Motivated by the potential of topology as a network resource, we ask: can the topology of a server-centric quantum data-center network support not only a distance-independent entanglement distribution rate, but also the distribution of end-to-end entanglement at \emph{hop-independent fidelity}? In a quantum data center, the fidelity required by a connection is determined by the application it serves, not by the physical distance or hop count between the requested QPUs. Ideally, the network should therefore hide its internal path length and provide an abstraction in which QPU--QPU connections offer comparable delivered entanglement fidelity. We formalize this objective by using the elementary-link quality as the benchmark: the fidelity of an end-to-end entangled state distributed across many hops should be restored to that of a successfully generated elementary link, independent of the number of hops. This is a stronger architectural target than satisfying any simple, fixed, end-to-end fidelity threshold. A fixed threshold may be application-specific, but it does not make a multi-hop segment scalable across modules. By contrast, when the fidelity of an entangled state distributed across many hops is restored to that of an elementary link, the network hides its internal physical distance; thus, from the perspective of a higher-level protocol, the segment can be treated as a single effective link. 

BCube offers a natural setting to test whether topology can deliver hop-independent fidelity, because its server-centric structure provides growing path diversity. Multiple edge-disjoint paths can be used to generate several raw end-to-end entangled states between the same pair of requested QPUs. These copies can then be consumed by entanglement purification to compensate for the fidelity loss caused by multi-hop swapping~\cite{Bennett1996Purification,Krastanov2019OptimizedPurification,Jansen2022Enumerating,Victora2023PurificationNetworks,Goodenough2024GraphCodes}. In this view, topology supplies the raw copy multiplicity, while purification converts that multiplicity into fidelity recovery.

A flat BCube network cannot be scaled-up indefinitely by relying on path diversity alone: increasing the BCube dimension increases the number of available redundant paths, but it also lowers the fidelity of each raw end-to-end copy entering purification. Once the raw input fidelity falls below the useful purification regime, additional redundancy cannot by itself guarantee recovery of link-level quality. This motivates a modular architecture in which each module acts as a fidelity-restoration block~\cite{Cirac1999Distributed,Shapourian2025QDC,Cacciapuoti2025QDC,pouryousef2026benchmarking}. If the internal paths and purification resources of a module can produce an end-to-end entanglement with similar rate and fidelity of the elementary-link, then the module preserves the resource scalability. Our assumed model is probabilistic elementary-link generation and deterministic swapping, with fidelity presented through the Werner model; such a module can be abstracted as an effective elementary link for a higher-level network. Figure~\ref{fig:qdc_network} presents the main idea of the modular quantum data-center architecture considered in this work. As we argue later in the paper, we adopt the BCube topology as our candidate for both the intra-rack and inter-rack networks, because it provides as many edge-disjoint paths as the path length and is therefore well suited to obtaining the hop-independent property. What then remains is whether this topology can generate enough useful copies to restore link-level fidelity while preserving the desired entanglement generation rate~\cite{Victora2023PurificationNetworks,Chen2024OptimumPurificationScheduling,Xiao2024PurificationScheduling}.

The required number of copies depends strongly on the purification protocol. We therefore compare two purification mechanisms. The Bennett \emph{et al.} protocol provides a canonical recursive 2-to-1 baseline~\cite{Bennett1996Purification}. We also study the higher-order purification family of Jansen \emph{et al.}, which enables nested $r$-to-1 purification based on bilocal-Clifford operations and can outperform repeated 2-to-1 purification in relevant regimes~\cite{Jansen2022Enumerating,Goodenough2024GraphCodes}. This comparison allows us to determine whether multi-copy purification can reduce the resource cost of hop-independent fidelity, and how the topology-level requirements depend on purification efficiency.

Building on this analysis, we present HI-QDC---Hop-Independent Quantum Data Center architecture. Then, we instantiate these requirements in a probabilistic BCube architecture and characterize the regimes in which a module can preserve elementary-link fidelity with an end-to-end success rate high enough that the recursive scale depth remains limited. Finally, we strengthen this condition by requiring the module to preserve both elementary-link fidelity and elementary-link yield rate. Taken together, these regimes identify modular operating points in which a purified BCube module hides its internal hop count and serves as an effective link at the next architectural scale. The code is available at~\cite{Azari2026HIQDCRepository}.

The remainder of the paper is organized as follows. Section~\ref{sec:blackbox} presents the black-box purification analysis and derives the resource requirements for restoring elementary-link fidelity over multi-hop paths. Section~\ref{sec:Architecture} details the design of HI-QDC---the Hop-Independent Quantum Data Center architecture. Section~\ref{sec:topology-aware} introduces the topology-aware BCube analysis, including probabilistic link generation, multiplexing, and the topology resource condition for preserving both fidelity and yield rate. Section~\ref{sec:discussion} discusses the architectural implications, limitations, and scalability interpretation of the results. Section~\ref{sec:conclusion} concludes the paper.

\begin{figure}[t]
    \centering
    \includegraphics[width=0.95\linewidth]{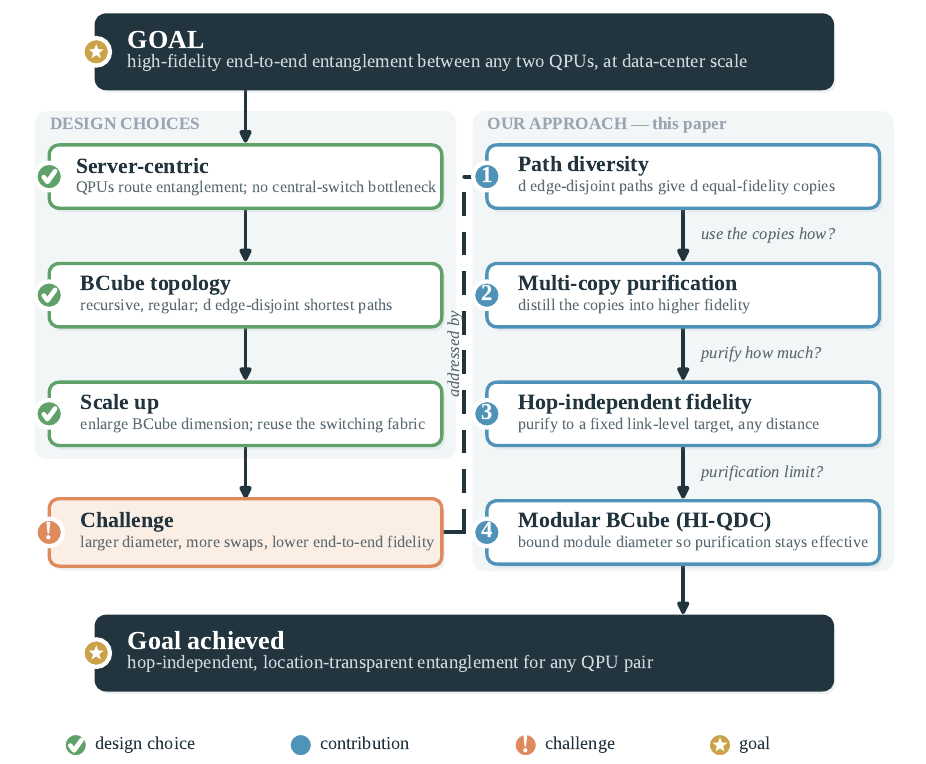}
    \caption{Road-map of the paper. A server-centric, recursive BCube fabric can be scaled to higher dimensions to provide data-center-scale connectivity, but this increased diameter lowers the raw end-to-end fidelity of cross-QPU connections. To overcome this limitation, we use the $d$ edge-disjoint shortest paths in BCube as $d$ equal-fidelity copies and apply multi-copy purification to restore the connection to a fixed link-level fidelity target (Section~\ref{sec:blackbox}). The HI-QDC organization (Section~\ref{sec:Architecture}) then bounds the diameter of each module, enabling hop-independent and location-transparent entanglement distribution between any pair of QPUs (Section~\ref{sec:topology-aware}).}
    \label{fig:three_stages}
\end{figure}

\section{Black-Box Analysis of Hop-Independent Fidelity}
\label{sec:blackbox}

In this stage, we isolate the purification requirement associated with multi-hop entanglement distribution before introducing any explicit network topology. Our goal is to determine whether the quality loss accumulated over a path can be reversed so that the resulting end-to-end state recovers the quality of a successfully established elementary link, a condition we refer to as hop-independent fidelity. This provides a topology-independent benchmark for the topology-aware analysis developed in section \ref{sec:topology-aware}.

\subsection{Motivation and Black-Box View}
\label{subsec:blackbox_motivation}

A central goal in a quantum data center network is to interconnect many quantum processing units (QPUs) in a way that remains useful as the network grows. As the network diameter increases, end-to-end entanglement must in general traverse a larger number of intermediate nodes, and the quality of the resulting end-to-end state degrades with hop count due to entanglement swapping at each intermediate node. Entanglement purification offers a possible way to reverse this degradation. However, for any fixed purification protocol and any fixed amount of link-level resources, there is a limit to how far this recovery can be sustained as the shortest-path distance increases. Understanding this tradeoff is therefore essential for assessing whether hop-independent fidelity is achievable at all, and at what resource cost.

This motivates the main question of this stage: can a multi-hop end-to-end entangled state be purified back to the quality of a successfully established elementary link, and if so, what resources are required? This question is shaped by several interacting factors, including the elementary-link quality, the number of hops, the purification protocol, and the number of copies available for purification. To understand the role of these factors more clearly, we first separate the purification problem from the topology problem. In this stage, we do not ask how a particular network topology generates multiple end-to-end copies, how often such copies occur, or how path redundancy depends on architecture. Instead, we adopt a black-box view of the network, in which the network is treated only as a provider of end-to-end entangled copies over a path of length $\ell$, and ask what purification alone can achieve with them. This abstraction isolates the purification requirement under the most favorable assumptions: if link-level fidelity cannot be recovered even in this setting, no topology can overcome that limitation. If recovery is feasible, the remaining question is whether a realistic topology can supply the required copies with sufficiently high probability, which we address in Section~\ref{sec:topology-aware}. In this sense, the black-box analysis establishes the purification requirements that any topology-aware solution must be able to meet.

\subsection{Model and Purification Framework}
\label{subsec:blackbox_model}

In the black-box setting, each successfully established elementary link is described by a Werner-state model, with Werner parameter $w \in [0,1]$,
\begin{equation}
\rho_W(w)= w \ket{\Phi^+}\bra{\Phi^+} + \frac{1-w}{4} I_4,
\label{eq:werner_state}
\end{equation}
where $\ket{\Phi^+}$ is a pure Bell state and $I_4$ is the $4 \times 4$ identity operator. Thus, the quality of an elementary link is fully characterized by a single parameter, with larger $w$ corresponding to higher entanglement quality. Throughout this analysis, the only source of imperfection is the initial Werner quality of the successfully generated elementary links. We do not incorporate additional noise from memory decoherence, storage-induced loss, or gate-level operational imperfections. This controlled setting is chosen deliberately so that the effect of hop accumulation can be isolated from other physical noise mechanisms.

We consider an end-to-end connection formed over a path of length $\ell$, where $\ell$ denotes the number of elementary links along the path between the two end QPUs. Since each hop is built from imperfect elementary links under ideal swapping, the Werner parameter of the raw end-to-end state before purification is $w_{\mathrm{e2e}}(\ell) = w^{\ell}$. The black box is therefore defined only through its input-output role: for a given elementary-link Werner parameter $w$ and path length $\ell$, it provides one or more raw end-to-end copies with Werner quality $w^{\ell}$, which purification then takes as its input resource. The success condition is that purification restores the end-to-end Werner parameter to at least $w^*$, where $w^* = w$ is the elementary-link Werner parameter, i.e., $w_{\mathrm{out}} \geq w^*$. This is precisely the condition of hop-independent fidelity: although the raw end-to-end quality decays with hop count as $w^{\ell}$, purification is considered successful if it restores the state to the link-level quality of the network regardless of path length. For a purification protocol that consumes $r$ input copies per round applied over $k$ nested recursion levels, meeting this target requires the existence of a fixed point satisfying $w_{\mathrm{out}}(w^*, r, k) \geq w^*$, ensuring that the protocol has a stable operating point at the target quality $w^*$.


To ensure that the conclusions are not tied to a single purification scheme, we consider two purification families. First, we use the Bennett \emph{et al.} protocol as the baseline, since it is the canonical 2-to-1 entanglement purification protocol and serves as the standard reference for recursive purification \cite{Bennett1996Purification}. Second, we consider the higher-order protocols of Jansen \emph{et al.}, which generalize purification to $r$-to-1 settings through bi-local Clifford distillation and were shown to outperform concatenated lower-order protocols in relevant operating regimes \cite{Jansen2022Enumerating}. The expression for the fidelity upgrade in terms of the Werner parameter $w$ and the success probability of the protocol for $r=4$ are as follows:
\begin{equation}
\begin{aligned}
p_{\mathrm{suc}} = \frac{3w^4 + 4w^3 + 1}{8}, \quad
w_{\mathrm{out}} = \frac{2w^2}{3w^2 - 2w + 1}.
\end{aligned}
\label{eq:4-to-1}
\end{equation}
For the complete success-probability and output-fidelity functions required for the multi-copy purification protocols used in this work, see Tables 5 and 6 in Appendix E of \cite{Jansen2022Enumerating}. Comparing these two families allows us to assess how strongly the resource requirements and feasibility of hop-independent fidelity depend on purification efficiency.

\subsection{Resource Scaling Under Recursive Purification}
\label{subsec:blackbox_results}

\begin{figure*}
    \centering
    \includegraphics[width=\linewidth]{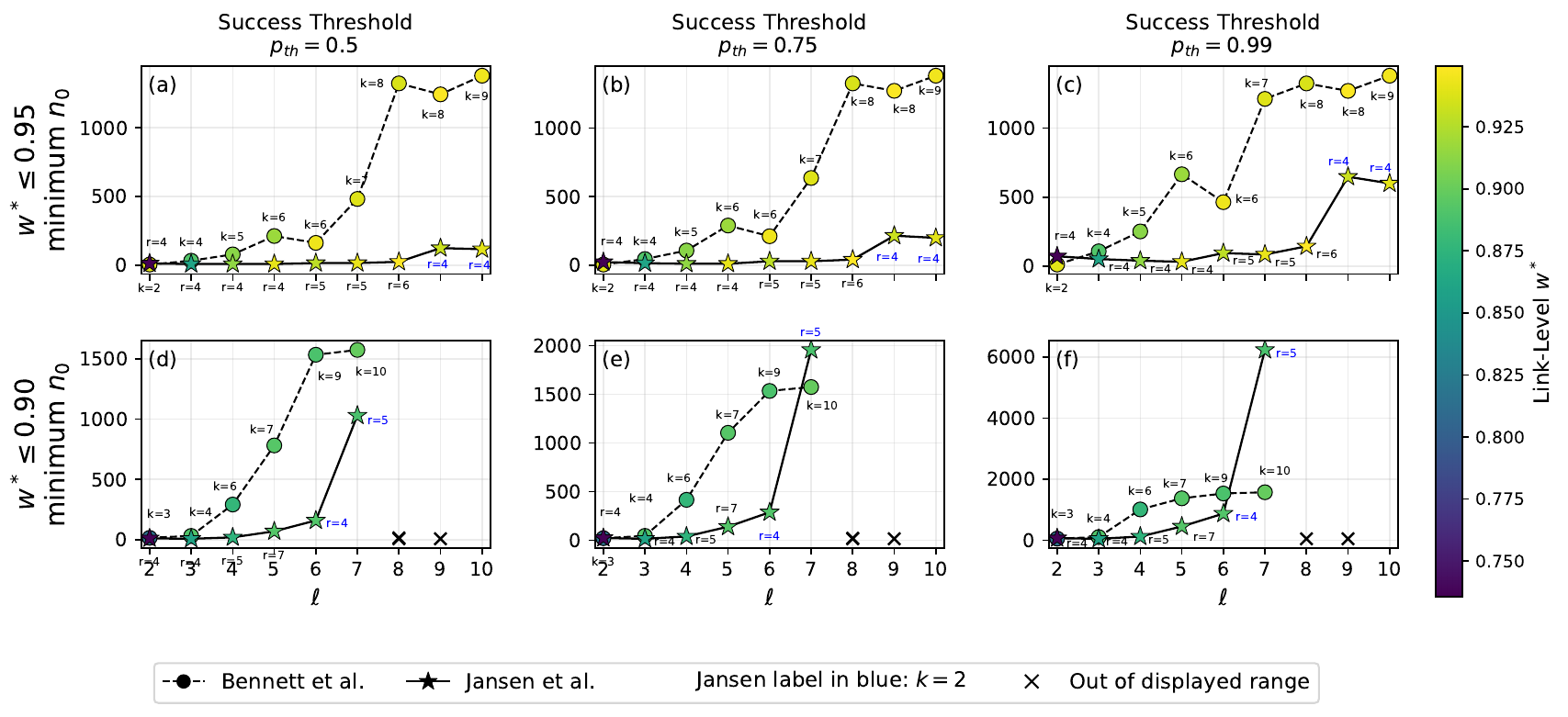}
    \caption{Minimum number of initial copies $n_0$ required to achieve hop-independent Werner parameter $w^*$ as a function of path length $\ell$ under different purification strategies. Columns correspond to success-probability thresholds $p_{th}$, and rows correspond to Werner parameter threshold $(w^*\leq w^{th})$. The y-axis shows the minimum $n_0$ such that $P_{\mathrm{succ}}\geq p_{th}$. Circles denote BBPSSW, and stars denote the best multi-copy protocol of Jansen \emph{et al.} selected from $r\in\{4,5,6,7\}$. Labels next to stars indicate the optimal $r$ and the optimal purification recursion depth $k$, with blue denoting $k=2$ and black denoting $k=1$. All points are computed using the smallest feasible purification depth $k$ for which purification achieves the target Werner parameter $w^*$.}
    \label{fig:fig3_best_jansen_vs_bbpssw}
\end{figure*}

\begin{figure*}
    \centering
    \includegraphics[width=\linewidth]{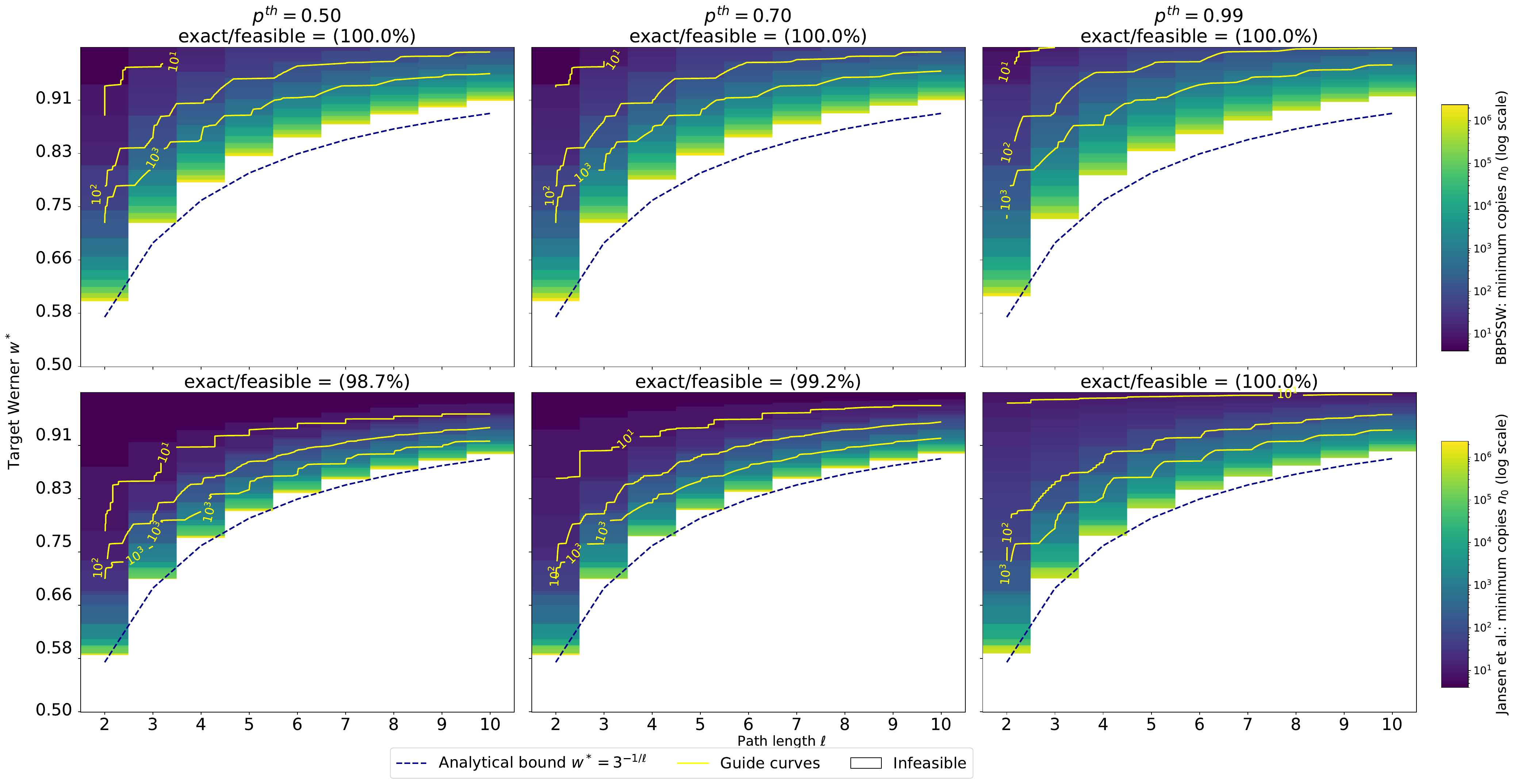}
    \caption{Black-box resource landscape for recovering a target Werner parameter $w^*$ over a path of length $\ell$ under different purification strategies and success-probability requirements. Each panel is a heatmap over $(\ell,w^*)$, where the color indicates the minimum number of input copies $n_0$ required to achieve $w_{\mathrm{out}}\geq w^*$ with success probability at least $p_{th}$. The top row corresponds to BBPSSW purification, while the bottom row shows the best Jansen $r$-to-$1$ protocol selected from $r\in\{3,4,5,6,7\}$. The dashed curve marks the analytical scaling $w^*=3^{-1/\ell}$, and the overlaid contours indicate constant $n_0$ values. White regions are infeasible within the explored search range.}
    \label{fig:blacbbox_heatmap}
\end{figure*}

We now characterize the resource cost of hop-independent Werner-parameter recovery by examining the minimum number of initial end-to-end copies $n_0$ required as a function of $w$ and $\ell$, under different purification protocols and success-probability thresholds. This comparison identifies the regimes in which hop-independent Werner-parameter recovery remains resource-feasible and the purification configurations that achieve it most efficiently.

Figure~\ref{fig:fig3_best_jansen_vs_bbpssw} provides a condensed view of this resource comparison. The horizontal axis gives the path length \(\ell\), while the vertical axis shows the minimum \(n_0\) required to recover hop-independent Werner parameter, subject to the imposed success-probability threshold. For each purification protocol, path length \(\ell\), and fixed-point constraint $w^* \leq w_{\mathrm{th}}$, where $w_{\mathrm{th}}$ is an upper bound on the target Werner parameter that parameterizes the link-quality regime under consideration, we identify the smallest recursion depth $k$ for which purification achieves the target Werner parameter $w^*$ and compute the minimum $n_0$ such that the overall purification success probability exceeds $p_{\mathrm{th}}$.

Several trends are visible in Fig.~\ref{fig:fig3_best_jansen_vs_bbpssw}. As the path length \(\ell\) increases, the input quality to purification degrades exponentially as \(w_{e2e}=(w^*)^\ell\), for any given $w^*$, which drives rapid growth in the minimum required \(n_0\), especially in the long-path regime. However, this growth is not strictly monotonic, as the system dynamically transitions between different purification depths $k$ and multi-copy configurations $r$, resulting in regime shifts that locally reduce resource requirements. Across most of the parameter space, multi-copy \(r\)-to-\(1\) purification significantly outperforms the standard BBPSSW protocol, particularly for moderate and large \(\ell\), although BBPSSW remains competitive for very short paths. Among the multi-copy protocols, \(r=4\) appears most often as the optimal choice, while higher-order \((r=5,6,7)\) become favorable in specific regimes where additional redundancy compensates for their increased purification cost. Moreover, relaxing the fixed-point constraint on the recovered Werner parameter, i.e., increasing \(w_{th}\), consistently reduces the required resources by improving the effective input to purification. By contrast, tightening the success-probability threshold primarily rescales the required \(n_0\) without changing the overall structure of the optimal purification strategy, indicating that the dominant bottleneck is governed mainly by fidelity constraints rather than purification success alone.

While Fig.~\ref{fig:fig3_best_jansen_vs_bbpssw} isolates the optimal operating point for each \(\ell\) under the imposed fixed-point constraints, Fig.~\ref{fig:blacbbox_heatmap} gives the more complete view of the black-box resource landscape by resolving the full \((\ell,w^*)\) domain. Each panel shows a heatmap in which the color indicates the minimum number of input copies \(n_0\) required to achieve \(w_{\mathrm{out}}\geq w^*\) with success probability at least \(p_{th}\). The dashed blue curve marks the analytical scaling \(w^*=3^{-1/\ell}\), which follows from the fundamental requirement that the input Werner parameter to purification must exceed $1/3$ for any protocol to improve fidelity, below which purification cannot recover link-level quality regardless of the number of copies available. The overlaid yellow contours provide reference levels for $n_0$, showing the growth of the resource requirement, while white regions indicate infeasibility within the explored search range. For each grid point, the smallest $n_0$ and corresponding protocol parameters satisfying both the fidelity constraint \(w_{\mathrm{out}}\geq w^*\) and the success-probability constraint \(P_{\mathrm{succ}}\ge p_{th}\) are selected by exhaustive evaluation over the discretized \((\ell,w^*)\) domain. This construction makes it possible to compare not only whether a point is achievable, but also how costly it is to achieve.

Figure.~\ref{fig:blacbbox_heatmap} shows that; first, tightening the success-probability threshold primarily affects the number of required input copies rather than the feasibility boundary itself. Second, the accepted regions exhibit a clean threshold structure in $w^*$: for each fixed path length \(\ell\), the feasible set forms a contiguous upper-tail interval characterized by a single lower boundary curve \(w^*_{\min}(\ell)\), which closely tracks the analytical scaling trend. Most importantly, the Jansen family consistently attains a larger accepted region with fewer required copies than BBPSSW, especially in the high-\(p_{th}\) regime, highlighting the advantage of multi-copy purification for scalable recovery of hop-independent Werner quality.

Taken together, the results show that resource-feasible hop-independent Werner recovery is governed primarily by path length and purification dynamics, that success-probability requirements act mainly as a resource multiplier rather than a feasibility limiter, and that multi-copy purification substantially enlarges the feasible operating region relative to BBPSSW.

\section{Architecture}
\label{sec:Architecture}

We use BCube as the structural substrate for the proposed quantum data-center architecture because it provides a clean way to connect topology, path diversity, and purification resources~\cite{Guo2009BCube,pouryousef2026benchmarking}. In a server-centric quantum data center, QPUs act both as communication endpoints and as intermediate repeaters. The main challenge is that multi-hop entanglement swapping degrades the fidelity of an end-to-end state as the path length increases. BCube is a useful test case for studying whether this degradation can be compensated by topology-provided redundancy, because its multipath structure has two important properties. First, between any two QPUs, the number of edge-disjoint shortest paths is equal to the length of the shortest path. Second, the switch levels, or equivalently the link types, traversed by these disjoint paths appear in cyclic order. Therefore, when all successful link-level entanglements generated at a given switch level have the same quality, the disjoint paths have identical end-to-end fidelity profiles.

This behavior is different from fixed-degree grid-like topologies. For example, the node degree is fixed at $4$ in a square grid and $6$ in a triangular grid, so the number of disjoint paths available to a QPU pair remains limited even as the network diameter grows. In such networks, longer paths increase the number of swapping operations without providing a corresponding increase in independent copies for purification. In contrast, BCube couples increasing path length with increasing path redundancy, making it a natural candidate for studying hop-independent entanglement distribution.

To make the cyclic path property concrete, consider the example shown in Fig.~\ref{fig:cyclic}. Here, QPUs A and B are separated by a shortest-path length of $3$, and the schematic shows one set of three edge-disjoint shortest paths between them. These paths traverse the same three link types in cyclically rotated orders---namely strides $\{1,2,3\}$, $\{2,3,1\}$, and $\{3,1,2\}$. For Werner states, the effective end-to-end Werner parameter after swapping is the product of the per-link Werner parameters along the path. Since multiplication is commutative, all three paths yield the same end-to-end Werner parameter whenever the per-level link qualities are uniform; that is, whenever all BSM stations at the same stride generate successful link-level entanglement with the same Werner parameter. The cyclic reordering changes which physical links are used, preserving edge-disjointness and improving fault tolerance, without changing the accumulated fidelity profile.

\begin{figure}[t]
\centering
\includegraphics[width=\linewidth]{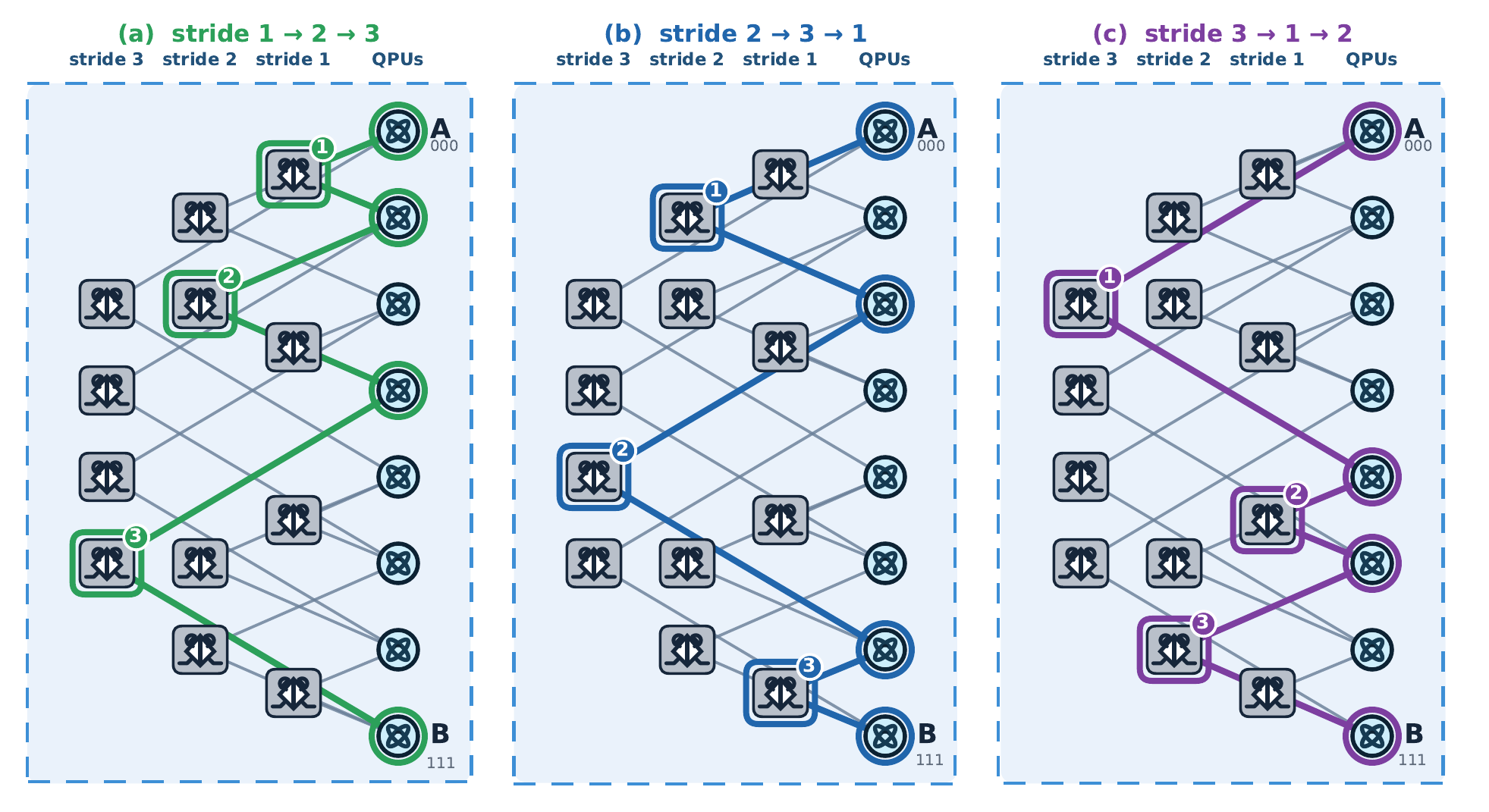}
\caption{Edge-disjoint paths between two QPUs separated by a length-3 shortest path, in a $\mathrm{BCube}(2,2)$ inner network. Each panel shows one of the three cyclic stride orderings. Highlighted nodes and edges mark the active path, and numbered badges indicate the traversal order. The three paths use the same three link types in cyclically rotated order, yielding the same end-to-end Werner-parameter product under uniform per-level link quality.}
\label{fig:cyclic}
\end{figure}

Formally, a $\mathrm{BCube}(n,k)$ contains $n^{k+1}$ QPUs addressed by $(k{+}1)$-digit strings over the alphabet $\{0,\dots,n{-}1\}$. It has $k{+}1$ switch layers, and two QPUs share their layer-$i$ switch precisely when their addresses agree in every coordinate except the $i$-th. Thus, a single hop through a layer-$i$ switch changes only coordinate $i$. In our quantum setting, these switch layers are realized by BSM stations, similar in role to the NIR switches considered in~\cite{Shapourian2025QDC}. With this structure, two QPUs at Hamming distance $d$ are joined by exactly $d$ edge-disjoint shortest paths, each of length $d$; see Theorem~\ref{thm:disjoint} in Appendix~\ref{sec:appendix}. The equality of their fidelity profiles follows from the cyclic ordering and commutativity argument above; see Theorem~\ref{thm:fidelity}.

These structural properties translate directly into entanglement-distribution resources. An $r$-to-$1$ purification step consumes $r$ independent end-to-end entangled state; in this setting, those copies must be prepared over $r$ disjoint end-to-end paths. Since the number of edge-disjoint shortest paths grows with the path length, larger-diameter BCube instances naturally provide more independent copies that can be used for purification, including multiple $r$-to-$1$ purification attempts in parallel. Moreover, because the disjoint paths accumulate the same set of per-level link qualities, the resulting end-to-end copies are well matched for purification at the end QPUs. Throughout this work, we model every BSM station as generating bipartite entanglement with success probability $p$ and, upon success, producing a Werner state with parameter $w$.

A conventional BCube can be scaled in two direct ways, but each direction introduces a different limitation. Increasing the order $k$ stacks $n$ copies of $\mathrm{BCube}(n,k)$ and connects them through a new layer of $n^{k+1}$ stations to form $\mathrm{BCube}(n,k{+}1)$, which contains $n^{k+2}$ QPUs. This \emph{scale-up} approach increases the network diameter to $k{+}2$ hops, so the worst-case number of swaps grows and the raw end-to-end fidelity decreases. Increasing the radix $n$ while holding $k$ fixed grows the QPU count as $n^{k+1}$, but each station must then serve $n$ QPUs. As a result, the per-station entanglement-generation load scales with $n$, and the stations can become bottlenecks, mirroring the centralized-resource limitation of switch-centric designs.

\subsection{The HI-QDC Architecture}

\begin{figure}
\centering
\includegraphics[width=\linewidth]{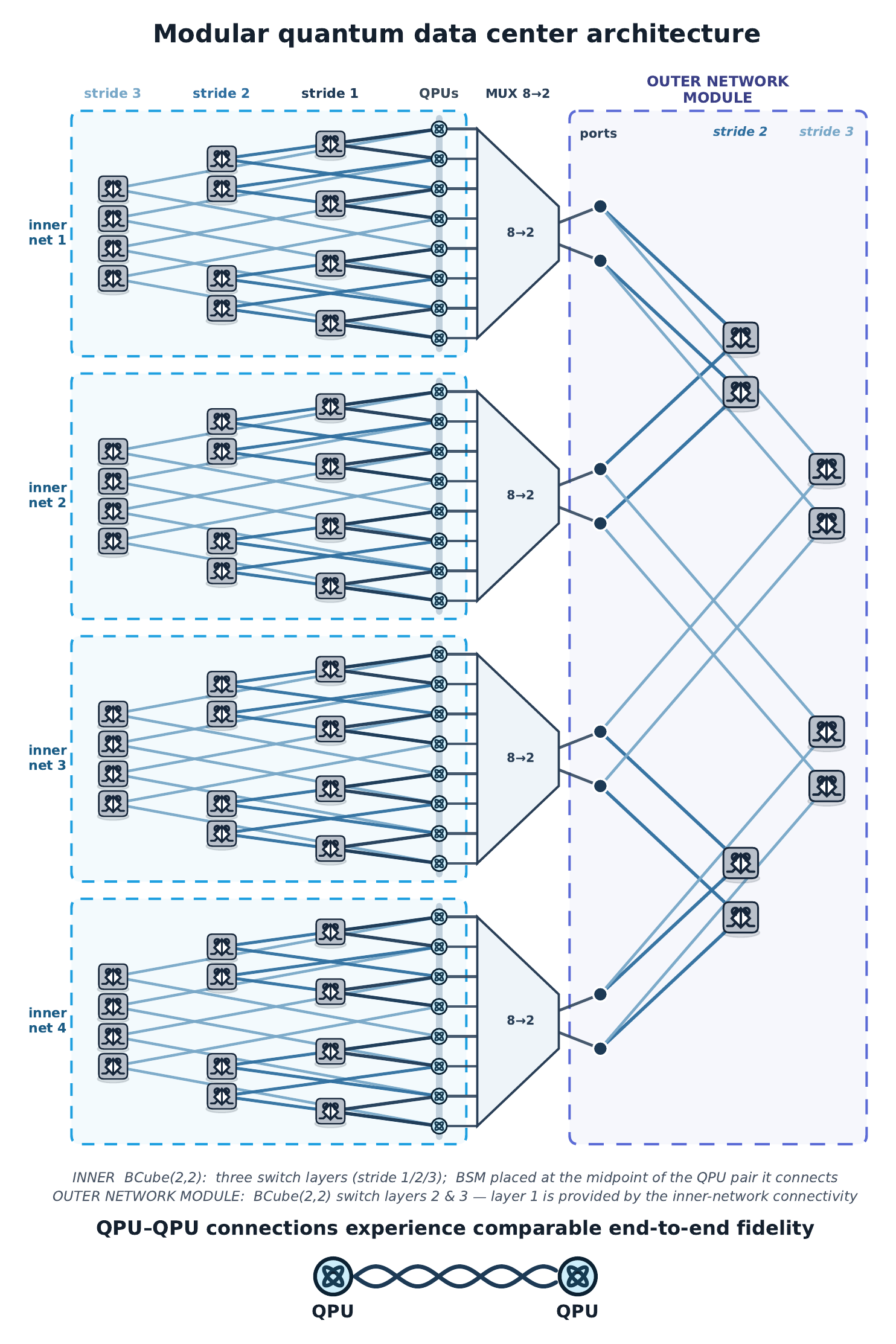}
\caption{Hierarchical entanglement distribution over a two-level $\mathrm{BCube}(2,2)$ architecture. The network contains $32$ QPUs connected through an outer $\mathrm{BCube}(2,2)$ network, while groups of eight QPUs are connected through inner $\mathrm{BCube}(2,2)$ modules marked by light-blue dashed boxes. Each inner module consists of three BSM-station layers, corresponding to strides $1$, $2$, and $3$, with four stations per layer. The outer network reuses purified intra-module end-to-end pairs as its stride-1 elementary links and therefore instantiates only its stride-2 and stride-3 BSM layers. Each inner module can distribute and purify an end-to-end pair between any two of its QPUs back into the accepted hop-independent resource region, so QPUs that connect to the outer network need not be physically adjacent within the inner topology. The outer network then performs inter-module routing, entanglement swapping, and purification so that QPU--QPU connections can attain the same target link-level Werner parameter regardless of their positions within or across modules.}
\label{fig:Architecture-recursive}
\end{figure}

We propose a third scaling mechanism, which we call HI-QDC: a recursive, isometric, modular scaling framework for quantum data centers that aims to distribute hop-independent entanglement across QPUs. Rather than increasing $k$ or $n$ directly, HI-QDC scales by composition. It nests BCube modules inside an outer module with the same structural form. We use the term \emph{isometric} to emphasize this self-similarity: every level of the recursion---inner and outer alike---is based on the same $\mathrm{BCube}(n,k)$ topology. Each inner module is a $\mathrm{BCube}(n,k)$ hosting $n^{k+1}$ QPUs, and $n^{k}$ such modules are interconnected by an outer module that is again a $\mathrm{BCube}(n,k)$. The total number of QPUs is therefore
\begin{equation}
\label{eq:total}
N \;=\; \underbrace{n^{k}}_{\text{modules}}\times\underbrace{n^{k+1}}_{\text{QPUs/module}} \;=\; n^{\,2k+1}. \end{equation}
For the $\mathrm{BCube}(2,2)$ instance shown in Fig.~\ref{fig:Architecture-recursive}, this corresponds to four inner modules, each containing eight QPUs, for a total of $N=32$ QPUs. Because the outer module reuses the same BCube structure, the construction can be applied recursively: each module at one level can itself be replaced by another HI-QDC block. Consequently, the hop-independent-fidelity condition only needs to be enforced locally at each level of the recursion.

The inner and outer networks are coupled through multiplexers. Each inner module presents its $n^{k+1}$ QPUs to an $n^{k+1}\rightarrow n$ multiplexer---the $8\rightarrow2$ stages in Fig.~\ref{fig:Architecture-recursive}---and thereby exposes $n$ outer-facing ports. Across all $n^{k}$ modules, the number of QPUs that can participate in cross-module connectivity at any instant is $n^{k}\times n=n^{k+1}$, exactly matching the endpoint count of the outer $\mathrm{BCube}(n,k)$. In this construction, the first layer of the outer BCube is supplied by intra-module connectivity, while the outer module realizes the remaining $k$ layers, namely layers $2$ through $k{+}1$. For the $\mathrm{BCube}(2,2)$ example, these are layers $2$ and $3$.

The remaining $n^{2k+1}-n^{k+1}$ QPUs participate in local entanglement through intra-module connectivity. This follows from the HI-QDC load model: at each protocol time step, each inner module can expose at most $n$ QPUs to the outer network, matching its $n$ outer-facing ports. For $\mathrm{BCube}(2,2)$, this corresponds to two outer-facing QPUs per inner module. This per-module allocation, however, need not be rigid. Because the defining goal of HI-QDC is to preserve the outer topology, the quantity that must remain fixed is the aggregate link-level capacity of the outer $\mathrm{BCube}(n,k)$, not the exact number of ports assigned to each module. Let $m_i$ denote the number of outer-facing QPUs drawn from module $i$. The symmetric allocation sets $m_i=n$ for every module, but this capacity can be redistributed according to the request pattern and the computational priority of the participating QPUs. Under full utilization, the conserved-capacity condition is
\begin{equation}
\label{eq:conservation}
\sum_{i=1}^{n^{k}} m_i = n^{k+1}.
\end{equation}
Thus, a module that needs more than $n$ of its QPUs to participate in the outer network can borrow capacity from modules that are not fully using theirs. The hard constraint is that the total outer-network capacity remains fixed, thereby preserving the topology and its associated fidelity guarantees.

The requirement that the inner and outer levels share the same fixed topology is tied to the baseline assumption that all BSM stations in the network provide homogeneous link-level entanglement resources. In this model, every BSM station distributes bipartite entanglement with success probability $p$ and, upon success, produces a Werner state with parameter $w$. Under this homogeneity assumption, each recursive level receives resources with the same statistical description, so the same topology-based redundancy argument can be applied at every level. If this assumption is relaxed, then the main design constraint becomes resource balancing. The link-level resources provided by the BSM stations, together with the redundancy provided by the topology, must compensate for both the fidelity loss caused by multi-hop swapping and the rate loss introduced by higher-order purification. A module can serve as an effective elementary link for the next level only if its output resources remain inside the hop-independent region required by the outer network.

\begin{figure*}[t]
\centering
\includegraphics[width=\linewidth]{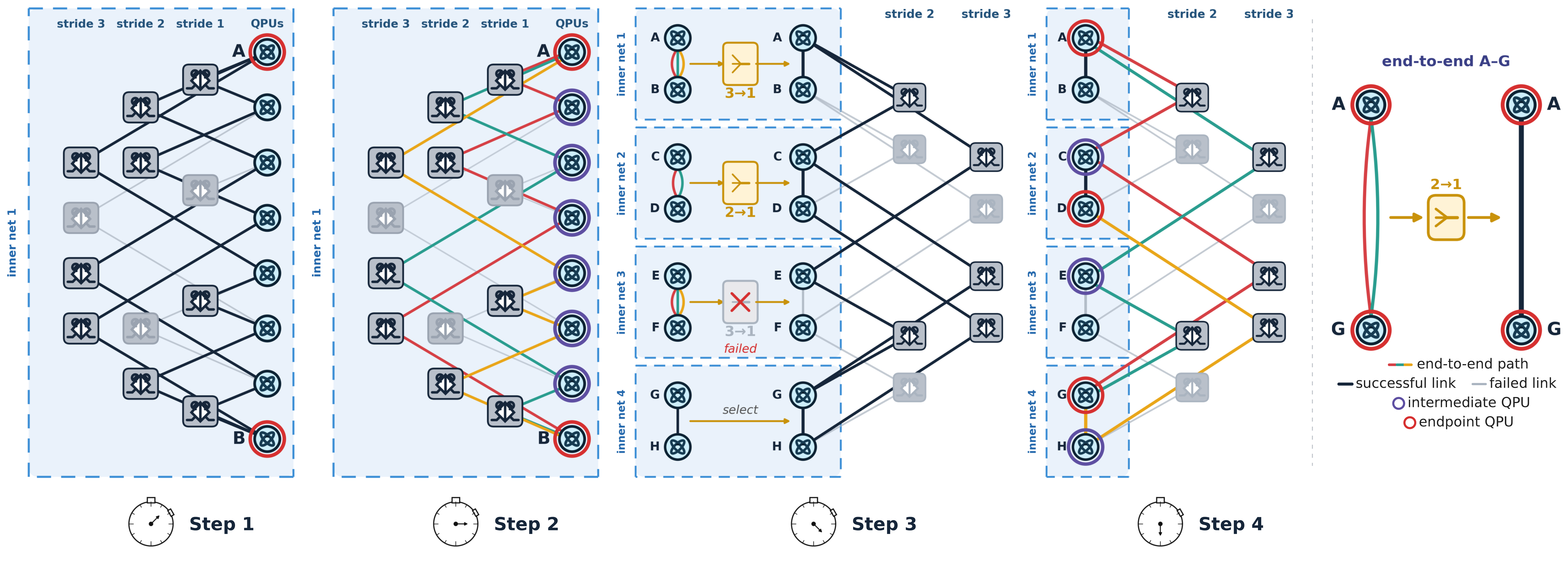}
\caption{End-to-end entanglement generation protocol in the proposed modular quantum data-center architecture over a $\mathrm{BCube}(2,2)$ topology. Step 1: Link-level Bell-state measurements are performed along strides $1$, $2$, and $3$ of an inner $\mathrm{BCube}(2,2)$ module, establishing raw entangled links between QPU pairs. Step 2: Three edge-disjoint paths between endpoints A and B are selected, and entanglement swapping is performed at intermediate QPUs to yield three independent end-to-end A--B pairs within the module. Step 3: The same preparation procedure is carried out independently and in parallel across all four inner modules. Modules with sufficient copies apply intra-module entanglement purification, indicated by the gold distillation box, while the others fall back to a single selected copy. Concurrently, the outer network attempts link-level entanglement generation across modules over strides $2$ and $3$. Step 4: The inter-module links are routed through the outer $\mathrm{BCube}(2,2)$ network, with entanglement swapping at intermediate QPUs, producing multiple end-to-end copies across modules. The final panel shows inter-module purification of the resulting copies, yielding a single high-fidelity end-to-end A--G entangled pair. Gray links indicate failed entanglement attempts or discarded copies; colored paths distinguish the edge-disjoint routes.}
\label{fig:walk-through}
\end{figure*}

We now walk through one protocol run for establishing entanglement between QPUs located in different modules. In this example, both the inner modules and the outer module follow a $\mathrm{BCube}(2,2)$ topology. Each inner module prepares a link-level-quality entangled pair that will act as one elementary link for the outer network. This preparation corresponds to the first two steps in Fig.~\ref{fig:walk-through}. In the first step, the inner modules attempt link-level entanglement generation across their BCube switch layers. In the second step, intra-module routing and entanglement swapping are performed along the successfully generated links in order to distribute end-to-end entanglement between the selected port QPUs.

We consider the worst-case intra-module setting, where the two QPUs selected as outer-network ports are not nearest neighbors but are separated by multi-hop paths. This can occur when the nearest QPUs are unavailable, busy with local tasks, or not selected to serve any request. Since the two selected QPUs in this example differ in all three BCube coordinates, they can be connected through up to three edge-disjoint shortest paths. For this farthest-apart pair, the union of the relevant shortest paths spans all link-level entanglement connections in the inner module. Therefore, during the first time step, all BSM stations in that module attempt elementary link generation across the BCube switch layers. During the second time step, routing and entanglement swapping are performed along the successfully generated paths, producing multiple raw end-to-end entangled pairs between the two port QPUs. If three edge-disjoint paths succeed, the inner module can then perform a $3$-to-$1$ purification step to convert these noisy multi-hop pairs into one higher-fidelity entangled pair.

Step~3 of Fig.~\ref{fig:walk-through} illustrates this preparation process across the inner modules. For example, the first inner module has enough successful end-to-end resources to perform purification and produce a higher fidelity entanglement, whereas another module may not have enough successful paths, as in inner module~4, or may have enough resources but fail during the purification procedure, as in inner module~3. A successful (purified) pair is then treated as a link-level entanglement resource for the outer network. During the final time step of this preparation phase, the outer module also attempts elementary link generation between ports of different modules using its own BSM stations. We should note that more generally, the strict isometric construction is the baseline used in HI-QDC. However, the same recursive argument can be extended to an outer network with larger $k$ or larger $n$, provided that the resources prepared by the inner modules remain in the hop-independent region required by that outer network.

After the link-level resources for the outer network have been prepared, the outer network performs routing and entanglement swapping to satisfy cross-module requests, as shown in Step~4. Suppose the requested connections are $\{A,G\}$ and $\{D,G\}$, and suppose $\{A,G\}$ has the highest priority. The controller would assigns $G$ to the outer-network port that provides the best route toward $A$. If that port is not the appropriate port for completing the route from $D$ to $G$, the module may use an additional intra-module pair, such as a pair between $G$ and an intermediate QPU $H$, to enable the second route. In this case, $H$ is not an end user of the final cross-module connection; it only serves as an intermediate QPU that exposes the required outer-facing port. This illustrates why some intra-module entangled pairs are needed even when one of their endpoint QPUs is not directly part of the user request.

Finally, after the outer-network paths are formed, the controller checks how many edge-disjoint end-to-end paths are available for each request. In the example shown, assume that only one end-to-end path is available for connecting $D$ and $G$, through the sequence in which $D$ shares entanglement with $H$ and $H$ shares entanglement with $G$. In contrast, two edge-disjoint paths are available for connecting $A$ and $G$: one through an intermediate QPU $C$ and another through an intermediate QPU $E$. The $\{D,G\}$ connection is therefore established directly through its single available path, whereas the $\{A,G\}$ connection can use a $2$-to-$1$ purification step across its two available edge-disjoint paths to improve the final cross-module fidelity. In this way, the HI-QDC architecture combines a quality-preserving intra-module topology with an isometrically recursive design to enable scalable, fidelity-preserving entanglement distribution across QPUs.

\subsection{Limitations and Benefits}

Fixing the topology (having the isometric modularity) comes with a limitation on simultaneous reach: at any protocol time step, only $n^{k+1}$ of the $n^{2k+1}$ QPUs can participate in cross-module connectivity. This fixed-topology constraint is adopted for two positive design advantages. First, because every module has the same $\mathrm{BCube}(n,k)$ structure, every level of the recursion has the same set of path lengths and therefore swap counts; with path lengths satisfying $d\in\{1,\dots,k{+}1\}$. This makes a uniform per-module purification budget well defined. The design then reduces to identifying the link-level resource quality and success probability (w,p)(w,p)
(w,p) that suffice, after routing, swapping, and purification over the edge-disjoint copies supplied by the topology, to deliver every QPU--QPU connection with hop-independent performance. In this work, hop independence means that the end-to-end resource should match the link-level benchmark, rather than merely improve over the unpurified multi-hop state.

This requirement is stronger than simply increasing fidelity after swapping. It asks whether the network can compensate simultaneously for the fidelity loss caused by longer paths and the rate loss introduced by purification. The role of the topology is therefore to determine how many independent end-to-end copies are available before additional multiplexing is introduced. In a non-multiplexed BCube, a distance-$d$ connection can draw on at most the $d$ edge-disjoint shortest paths supplied by the topology, so the purification copy budget is tied directly to path length. Multiplexing can enlarge this budget by generating multiple parallel entanglement at an elementary link, thereby providing additional independent copies even for shorter paths. Thus, the central architectural question is whether the combination of topological redundancy and multiplexing can supply enough copies to restore the link-level resource without exhausting the rate budget.

Second, once such a hop-independent operating regime exists for one module, the isometric recursion propagates it across scales. Let $\mathcal{M}$ denote the effective map induced by one $\mathrm{BCube}(n,k)$ module, taking a link-level resource with parameters $(w',p')$ to the corresponding end-to-end resource after routing, swapping, and purification. If there exists a regime in which
\begin{equation}
\mathcal{M}(w',p') \geq (w',p'),
\end{equation}
then the module outputs an end-to-end resource that is at least as good as the resource it received at its input. This output can therefore serve as the link-level resource for the next recursive level. Because the next level has the same topology, it applies the same map $\mathcal{M}$ and again returns a resource no worse than $(w',p')$. The argument then repeats unchanged at every level of the hierarchy.

Therefore, a single demonstration of the hop-independent regime for one $\mathrm{BCube}(n,k)$ module is sufficient to establish the same guarantee for every module at every recursion depth, provided that the same resource assumptions are maintained. Hop independence is thus invariant under the HI-QDC scaling construction, which is the central architectural guarantee of the design.

Beyond this guarantee, the fixed topology also proves resource-efficient, requiring substantially fewer physical resources than an equivalent flat network. If the same number of QPUs were hosted by a flat BCube, the network would need to scale to $\mathrm{BCube}(n,2k)$, which requires substantially more BSM stations. In contrast, HI-QDC reuses smaller $\mathrm{BCube}(n,k)$ modules and connects them through an outer module of the same form. As shown in Appendix~\ref{sec:appendix}, this reduces the BSM count by a factor approaching one half as $k$ grows. For the $\mathrm{BCube}(2,2)$ example, HI-QDC uses $56$ BSM stations instead of the $80$ required by a flat $\mathrm{BCube}(2,4)$. Thus, the restricted simultaneous connectivity of HI-QDC naturally brings a significant reduction in physical switching resources.

\section{Topology-Aware Resource Analysis}
\label{sec:topology-aware}

Building on the HI-QDC architecture introduced in Section~\ref{sec:Architecture}, we now turn from the structural design to the resource question it raises. That section established why BCube is a useful substrate: although every entanglement swap along a multi-hop path degrades fidelity, the topology supplies as many edge-disjoint shortest paths as the distance separating a QPU pair, and these paths provide the independent end-to-end copies that purification consumes. The goal of this section is to quantify when those topology-provided copies suffice to offset multi-hop fidelity degradation, and when additional link-level resources or multiplexing are required.

We study this question in two regimes. In the first regime, we consider a fidelity-focused version of the modular HI-QDC design. Here, purification is used to restore the end-to-end Werner parameter of each module back to the link-level target, but the probability of successfully producing such a purified resource may still decrease as the recursion depth increases. This gives a limited form of hop-independent fidelity: the output state quality can be maintained across hops, but scalable operation is possible only up to the point where the end-to-end success probability remains above an acceptable threshold.

In the second regime, we impose a stronger requirement needed for scalable recursive operation. It is that in addition to restoring fidelity, the network must also maintain a hop-independent rate, meaning that the end-to-end success probability or yield must match the corresponding link-level benchmark. This condition is necessary if hop-independent fidelity is to remain scale-free across modules. In other words, a hop-independent rate is a prerequisite for hop-independent fidelity. We therefore study the minimum link-level resources, specified by the Werner parameter $w$, the elementary-link success probability $p$, and the available multiplexing $m$, that allow a $\mathrm{BCube}(n,k)$ module to reproduce its input resource at its output. Together, these two regimes separate scale-limited fidelity-only recovery from true scalable hop-independent performance and clarify the additional resource conditions required for the latter.

\subsection{Maintaining Link-Level Fidelity}

\begin{figure*}[t] 
\centering 
\includegraphics[width=\linewidth]{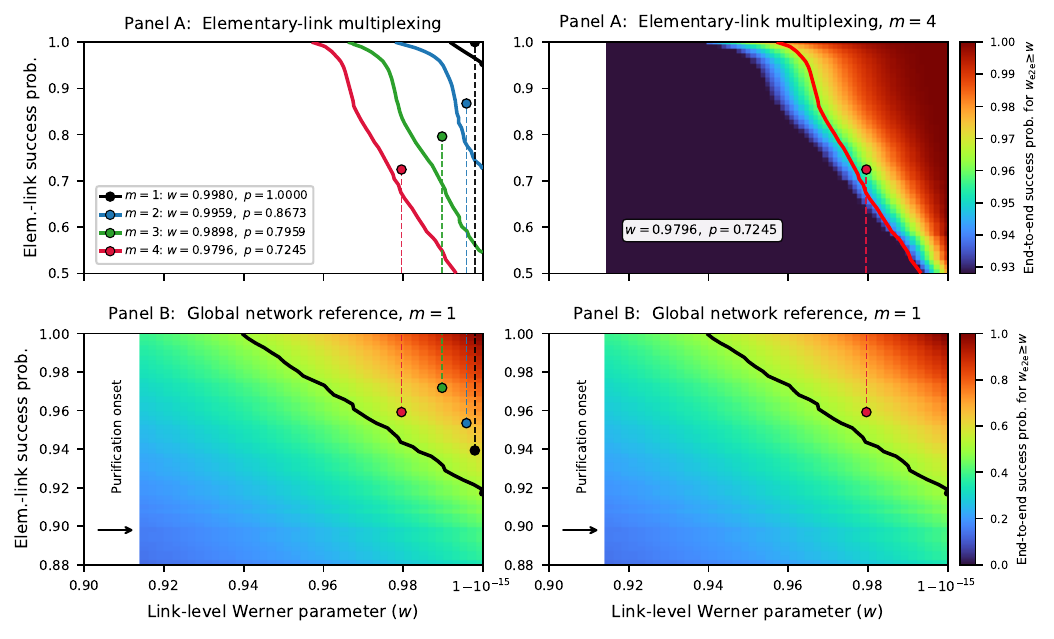}
\caption{Accepted resource region for hierarchical entanglement distribution in a $\mathrm{B\text{-}Cube}(2,3)$ network. The horizontal axis is the link-level Werner parameter $w$, and the vertical axis is the link-level entanglement generation success probability $p$. Heatmap color shows the end-to-end success probability of distributing entanglement with $w_{\mathrm{e2e}} \geq w$ between diameter-apart hosts. In the top panel, contours mark the smallest elementary-link multiplexing level $(m)$ for which the inner-network output can serve as an effective link of an outer $\mathrm{B\text{-}Cube}(2,3)$ with end-to-end success probability at least $0.5$. The bottom panel shows the corresponding outer-network reference for $m=1$.}
\label{fig:heatmap_limitted_hop_fid} 
\end{figure*}

This subsection studies a first practical form of hop-independent fidelity in BCube, which we refer to as \emph{limited hop-independent fidelity}. As discussed in Section~\ref{sec:Architecture}, a $\mathrm{BCube}(n,k)$ network has diameter $k{+}1$, so the longest shortest path between any two QPUs traverses $k{+}1$ elementary links and therefore undergoes $k$ entanglement swaps. The strongest possible target for hop-independent fidelity would be to make the network behave, from the standpoint of fidelity, as if every QPU--QPU connection had effective length one, so that any pair of QPUs distributes entanglement with an end-to-end Werner parameter equal to the elementary link-level value, $w_{\mathrm{out}}\geq w$. This objective is demanding in practice, because the fidelity of the swapped state degrades with every swap. The number of edge-disjoint end-to-end copies the topology supplies, however, grows with the length of a connection, and faster than the number that purification requires, so the longest connections carry more than enough copies while the shorter multi-hop pairs fall short. We therefore begin with a more practical target: determining whether the farthest-apart QPU pairs can recover the link-level Werner parameter while still maintaining a meaningful end-to-end entanglement generation rate.

Under this target, we examine whether BCube supplies a sufficient route budget, that is, enough edge-disjoint paths to provide the input copies that purification requires at each path distance. The black-box purification analysis shows that paths of length $2$, $3$, and $4$ require at least $3$, $4$, and $4$ input copies, respectively, to recover the elementary link-level Werner parameter over a nontrivial resource region under the $r$-to-$1$ purification protocol of Jansen \emph{et al.}~\cite{Jansen2022Enumerating}. On the supply side, $\mathrm{BCube}(n,k)$ provides exactly $d$ edge-disjoint shortest paths between a pair separated by distance $d$, and each path delivers at most one independent end-to-end copy per non-multiplexed round; under probabilistic entanglement generation, in which some of these paths may fail, at most $d$ copies are available in a given round. A diameter-apart pair in $\mathrm{BCube}(n,1)$ therefore has distance $2$ and supplies at most two copies, whereas at least three are needed, and a diameter-apart pair in $\mathrm{BCube}(n,2)$ has distance $3$ and supplies at most three copies, whereas at least four are needed. The first feasible case is $\mathrm{BCube}(n,3)$: its diameter-apart pairs have distance $4$ and access four edge-disjoint copies, matching the minimum requirement of the $4$-to-$1$ purification protocol. Thus $\mathrm{BCube}(n,3)$ is the smallest BCube instance in which purification-based diameter reduction becomes feasible without relying on additional multiplexing.

This property is already significant within a single network. Restoring the fidelity of diameter-apart pairs reduces the effective network diameter even when the same restoration is not achieved for every QPU--QPU pair. In $\mathrm{BCube}(2,3)$, applying $4$-to-$1$ purification to the diameter-apart pairs upgrades the original length-$4$ connections to effective length one. Using the QPU-pair distance distribution of $\mathrm{BCube}(2,3)$, once these distance-$4$ pairs are upgraded the distance-$3$ pairs can reroute through an upgraded length-$1$ connection, so that the remaining two-thirds of QPU--QPU pairs become effectively bounded by length $2$. Even this limited form of hop-independent fidelity therefore produces a nontrivial reduction in the effective network diameter.

The stronger question is whether this restored diameter-apart resource can support further scalability. To address it, we consider the HI-QDC architecture in which eight $\mathrm{BCube}(2,3)$ inner networks are interconnected so that selected QPU--QPU connections within each module serve as effective links of an outer network. Through the multiplexers of HI-QDC, each inner module exposes two of its QPUs to the outer fabric, so the eight modules form a sixteen-node outer network with the same $\mathrm{BCube}(2,3)$ topology. The resulting hierarchical network contains
\begin{equation}
n^{k}\times n^{k+1}=2^{3}\times2^{4}=128
\end{equation}
QPUs in total. We take as outer-network elementary links those inner-network connections that restore $w_{\mathrm{e2e}}\geq w$ and also meet a sufficiently high entanglement generation rate. Under this condition, we identify the parameter region in which the outer network can distribute entanglement across its own diameter at effective length one, with the same fidelity as the inner link-level value, while keeping the end-to-end generation success probability above $0.5$. This region encompasses $54\%$ of all end-to-end QPU--QPU connections; the remaining connections are those whose inner-network distance is already $2$, together with the new length-$2$ inter-network connections. Diameter-level Werner-parameter restoration is therefore not merely a local improvement within a single module, but a resource for larger-scale recursive entanglement distribution.

The recursive construction follows the HI-QDC structure illustrated in Fig.~\ref{fig:Architecture-recursive}, which shows a smaller $\mathrm{BCube}(2,2)$ instance for visual clarity; the quantitative results in this subsection are obtained for recursively connected $\mathrm{BCube}(2,3)$ modules. The global objective of the HI-QDC architecture was for QPU pairs to experience an end-to-end Werner parameter equal to that of a length-$1$ or few-hop connection, regardless of which QPU within a module participates in the request. This cannot be achieved by exposing only nearest-neighbor, length-$1$ inner connections to the outer network: such connections already satisfy $w_{\mathrm{e2e}}=w$ without purification, but they make only adjacent QPUs available for cross-module communication, whereas the QPUs that require cross-network entanglement are set by the workload rather than by physical adjacency. The adjacent QPUs may already be occupied with local computation, and several mutually non-adjacent QPUs in the same module may need outer-network connections simultaneously. Relocating an arbitrary QPU state to a port-adjacent QPU would itself require teleportation, hence another inner entangled link, reintroducing the very multi-hop connection it was meant to avoid. Inner purification is therefore not an avoidable overhead; it is the mechanism that lets any QPU, and several at once, act as outer-network endpoints at a uniform link-level Werner parameter.

The quantitative result is summarized in Fig.~\ref{fig:heatmap_limitted_hop_fid}, which characterizes the accepted region of link-level Werner parameter $w$ and elementary-link entanglement generation success probability $p$ for a $\mathrm{BCube}(2,3)$ topology. The accepted region consists of those link-level resources for which $r$-to-$1$ purification at the diameter-apart hosts raises the end-to-end Werner parameter to at least the link-level value, $w_{\mathrm{e2e}}\geq w$, while also providing an end-to-end success probability high enough for the purified inner-network connections to serve as effective links of an outer network with the same topology. This recursive acceptance criterion is strictly stronger than Werner restoration alone: the end-to-end rate of the diameter-apart inner connections must remain high enough that, when they are reused as outer-network links, the outer network can still distribute entanglement between its own diameter-apart host pairs, the furthest-apart pairs, with success probability at least $0.5$. Because the outer network inherits the same effective link-level Werner parameter from the inner-network outputs, the only remaining requirement for recursive scalability is that this inherited rate be sufficient.

Figure~\ref{fig:heatmap_limitted_hop_fid} is arranged as two columns, each with an upper panel~A and a lower panel~B. In all cases, every pair of neighboring hosts performs $m$ parallel entanglement generation attempts that succeed independently with probability $p$, and upon success each entangled state is defined as a Werner state with parameter $w$. For each $(w,p)$ we compute the end-to-end success probability of generating at least one entangled state between diameter-apart hosts with Werner parameter at least $w$ through $r$-to-$1$ purification. In each run, $r$ is set to the smallest value that achieves hop-independence at the given path length. When the number of available entangled pairs exceeds $r$, the surplus is used to run multiple purification attempts in parallel. For $\mathrm{BCube}(2,3)$ with $m=1$, $r=4$ contributes to the end-to-end rate, because a length-$4$ path requires at least $4$-to-$1$ purification to restore its output Werner parameter to the link-level value. Panel~A of the left column overlays, for each $m\in\{1,2,3,4\}$, the contour above which the purified inner-network output qualifies as an outer-network link, so that the accepted region for a given $m$ lies above its contour; the marked points give the representative $(w,p)$ values carried into panel~B.

Panel~B is the global, outer-network reference for $m=1$: its color encodes the end-to-end success probability of the outer network across its own diameter, and the mapped points show where each inner-network output lands, at coordinates $(w_{\mathrm{e2e}},p_{\mathrm{e2e}})=(w,P_{A})$ set by the restored Werner parameter and the inner-network success probability $P_{A}$. The right column isolates $m=4$ so that the panel~A color is free to represent the inner-network end-to-end success probability itself, with its accepted-region boundary drawn in red; this makes visible how the expected number of available end-to-end paths combines with the per-purification success probability to form the contour. The purification onset annotated in panel~B marks the threshold in $w$ below which purification cannot raise the end-to-end Werner parameter to link-level quality. In this way each value of $m$ maps an accepted region of inner-network link-level resources onto a corresponding accepted region for the outer network; smaller multiplexing values produce a markedly narrower mapped region, so increasing the number of multiplexed attempts substantially enlarges the regime in which recursive end-to-end performance can be sustained.

In the Monte Carlo data, the contour along which the outer network maintains an end-to-end success probability above the target threshold partitions the link-level resource plane into two regimes. Reading the link-level Werner parameter $w$ and the link-level success probability $p$ as interchangeable resources, the lower-$p$ branch of the contour is close to linear and reflects a supply-limited operating point: too few elementary links arrive to populate the parallel groups, and the resulting deficit is offset by raising $w$, so that classical success probability and fidelity are exchanged at an approximately constant rate. Because this rate is steep, $w$ is the scarcer resource, and a modest reduction in fidelity must be repaid by a large gain in $p$. The upper-$p$ branch is sharply nonlinear and nearly vertical, corresponding to a purification-limited regime in which copies are plentiful but the per-purification success is too small for the $m$ parallel purification to reliably deliver a restored link; there, additional $p$ yields little benefit and the boundary is fixed almost entirely by $w$. Increasing the multiplexing degree $m$ translates the entire contour toward smaller $w$ and smaller $p$, so the same global target is reached with strictly fewer resources of either kind. This follows from the structure of the inner-network success probability, in which the number of copies available for purification grows with $m$: as $m$ increases, the probability that at least one group is distilled successfully saturates, the fidelity bottleneck recedes, and the purification-limited wall contracts toward lower $w$ while the linear supply-limited arm expands.

Cast as a resource-allocation problem, the three quantities a designer can provision at a fixed topology, namely the link-level Werner parameter $w$, the link-level success probability $p$, and the multiplexing degree $m$, are far from interchangeable in meeting the dual requirement that the outer network deliver hop-independent end-to-end entanglement, $w_{\mathrm{e2e}}\geq w$, at an end-to-end success probability above the chosen threshold. The link-level entanglement fidelity is the dominant and non-substitutable resource: it governs the purification acceptance, and through the hop-independence onset it imposes a hard floor, independent of $p$ and $m$, below which the restored Werner parameter cannot reach $w$ at all, so that no amount of supply or parallelism can compensate for a link whose fidelity is too low. Its marginal effect on the end-to-end success probability is correspondingly the steepest, more than an order of magnitude larger per unit than that of $p$, yet in the relevant regime fidelity already lies close to unity and leaves little headroom to spend. The success probability and the multiplexing degree instead act as substitutable supply, both increasing the number of distilled entangled pairs delivered to the outer network and therefore trading against one another in meeting the success-probability threshold; multiplexing provides the larger gain while $m$ is small but yields diminishing returns as it grows, whereas $p$ is the gentler lever that retains the most headroom. Ranked by physical effect per unit added, and setting provisioning cost aside, the priority is $w$, then $m$ while $m$ remains small, then $p$; the economically optimal allocation further depends on the cost of each resource, which we consider next.

Increasing the elementary-link success probability $p$ requires physical-layer improvements such as lower loss, higher coupling efficiency, better detectors, and more efficient Bell-state measurements. Increasing the elementary-link fidelity requires either better physical noise suppression or post-generation purification, the latter consuming multiple successful raw pairs and adding memory, operation, and scheduling overhead. Increasing the multiplexing degree $m$, by contrast, is primarily an architectural cost: it adds parallel modes, ports, memories, or detectors to generate more raw copies without requiring any individual physical attempt to improve. Multiplexing is therefore often the most accessible engineering knob, fidelity improvement the most resource-consuming, and an increase in $p$ falls between the two as a physical-layer hardware improvement.

Overall, this analysis identifies the resource requirement for the restoration of the link-level Werner parameter on diameter-apart entanglement distribution as the first non-trivial step toward hop-independent fidelity in BCube networks. More importantly, it shows that Werner restoration alone is not sufficient for recursive scalability: the restored connections must also preserve a strong enough end-to-end entanglement generation probability to function as reliable resources at the next network level.

\subsection{Toward Hop-Independent Fidelity}

Scalable hop-independent fidelity requires more than restoring the end-to-end Werner parameter: it requires preserving both end-to-end entanglement quality and end-to-end rate across network modules, i.e., both \(w_{e2e}=w\) and \(p_{e2e}=p\). To summarize our findings so far, B-Cube is the strongest candidate for this objective because its end-to-end path redundancy grows with Hamming distance. This growth in edge-disjoint shortest paths provides the redundancy needed for purification over long distances, while HI-QDC recursive architecture provide the natural mechanism for extending hop count as the feasible resource region shrinks with network size. The recursive study of the smallest feasible case, B-Cube\((2,3)\), already showed that restoring the Werner parameter alone is not sufficient for unbounded scaling. The central question of this subsection is therefore whether scale-invariant hop-independent fidelity is achievable at all, and if so, what initial link-level resources, in terms of Werner parameter \(w\), entanglement generation rate \(p\) and, multiplexing \(m\) are required.

\begin{table}[t]
\centering
\caption{For each $k\le 10$ in the non-multiplexed case ($m=1$), with the corresponding minimum feasible purification size $r$, we check whether there exists any $w\in[w_{\min},1)$ such that a nontrivial $p\in(0,1)$ satisfies $R_{\mathrm{out}}\ge p$. If such a region exists, the table lists the first scanned $w$ and the first scanned $p$ where this occurs.}
\label{tab:bcube_anyw_nontrivialp_k10}
\begin{tabular}{c c c c c}
\hline
$k$ & $r$ & $w_{th}$ & $\exists\,\min{w}$ & $\exists\,\min{p}$ \\
\hline
3  & 4 & 0.913743 & --        & --        \\
4  & 4 & 0.941308 & 0.9999317 & 0.9977339 \\
5  & 4 & 0.957439 & 0.9996385 & 0.9892078 \\
6  & 4 & 0.967703 & 0.9993895 & 0.9810882 \\
7  & 4 & 0.974642 & 0.9990040 & 0.9793058 \\
8  & 4 & 0.979555 & 0.9982222 & 0.9909239 \\
9  & 4 & 0.983163 & 0.9974163 & 0.9886137 \\
10 & 4 & 0.985891 & 0.9969899 & 0.9862372 \\
\hline
\end{tabular}
\end{table}

We begin with the non-multiplexed case, \(m=1\), and ask whether a diameter-apart end-to-end connection can simultaneously preserve both the elementary-link Werner parameter and the elementary-link rate. Table~\ref{tab:bcube_anyw_nontrivialp_k10} reports, for each feasible network size, the minimum link-level Werner parameter \(w\) and the minimum link-level entanglement generation success probability \(p\) for which diameter-apart host-host connections satisfy both conditions. A first key observation is that feasibility begins at \(k=4\), which identifies B-Cube\((n,k)\) with \(k\geq 4\) as the first nontrivial regime in which unbounded-scale hop-independent fidelity is not ruled out by the probabilistic nature of purification. A second key observation is that the minimum required purification size remains \(r=4\) for all feasible cases up to \(k=10\), while the minimum required link-level Werner parameter decreases as the network size increases.

This monotonic decrease in the required \(\min w\) is a direct consequence of the interaction between topology and purification. For diameter-apart pairs in B-Cube, increasing \(k\) increases the number of available edge-disjoint shortest paths, up to \(k+1\), and these parallel end-to-end copies are precisely the resources consumed by \(r\)-to-\(1\) purification. Although the raw end-to-end fidelity decreases with path length, restoring the link-level Werner parameter still requires only a fixed \(r=4\) purification, and the additional edge-disjoint paths in larger networks supply more purification inputs, so the end-to-end Werner condition can be sustained even at smaller link-level Werner parameters. However, the requirement imposed in Table~\ref{tab:bcube_anyw_nontrivialp_k10} is stronger than Werner restoration alone: the end-to-end rate must also match the link-level rate. This is why the reported minimum \(w\) values remain substantially higher than the Werner threshold \(w_{\mathrm{th}}\), which captures only the fidelity condition imposed by network diameter through purification.

\begin{table*}[t]
\centering
\caption{For each feasible $k\le 10$, the minimum purification size is $r=4$. The table shows the threshold $w_{\mathrm{th}}$, the minimum multiplexing level $m$ for which the yield criterion is feasible, together with the corresponding minimum $w$ and minimum $p$ at that minimum $m$. It also shows the minimum $w$ and minimum $p$ for the fixed maximum multiplexing level $m=5$.}
\label{tab:bcube_yield_minm_maxm_ieee}
\begin{tabular}{c c c c c c c}
\hline
$k$ & $w_{\mathrm{th}}$ & $\min m_Y$ & $\min w_Y$ & $\min p_Y$ & $w_Y|_{m=5}$ & $p_Y|_{m=5}$ \\
\hline
3  & 0.9137430540 & -- & --           & --           & --           & --           \\
4  & 0.9413085420 & 2  & 0.9996666660 & 0.9913320000 & 0.9882858721 & 0.9964759990 \\
5  & 0.9574399513 & 2  & 0.9770325457 & 0.9998599990 & 0.9811885506 & 0.9944519996 \\
6  & 0.9677033094 & 2  & 0.9802219856 & 0.9929879997 & 0.9770923165 & 0.9925479997 \\
7  & 0.9746423679 & 2  & 0.9746423679 & 0.9954359995 & 0.9746423679 & 0.9900000000 \\
8  & 0.9795557016 & 2  & 0.9795557016 & 0.9785114047 & 0.9795557016 & 0.9709843939 \\
9  & 0.9831634048 & 2  & 0.9831634048 & 0.9654381755 & 0.9831634048 & 0.9583073232 \\
10 & 0.9858912905 & 2  & 0.9858912905 & 0.9567226894 & 0.9858912905 & 0.9499879956 \\
\hline
\end{tabular}
\end{table*}

The corresponding behavior of the minimum required \(p\) reflects two competing effects. On one hand, increasing \(k\) improves the probability of obtaining enough redundant end-to-end copies for purification, because more edge-disjoint shortest paths become available. On the other hand, increasing \(k\) also lowers the quality of the purification inputs, since the states entering purification are distributed over longer paths, and it further reduces the success probability of the purification step itself. These competing effects explain why the minimum required \(p\) initially decreases with increasing \(k\), reflecting the growing topological redundancy, and then increases again after approximately \(k=7\), where the reduced input quality begins to dominate. Thus, even in the non-multiplexed setting, the feasible operating region is governed by a balance between topology-enabled redundancy and the increasing difficulty of purifying lower-quality long-distance states.

We next ask whether increasing the link-level resource population through parallel attempts can relax resource requirements. Operationally, this requires multiple memories to store successful entanglement and Bell-state-measurement stations capable of supporting multiple attempts in parallel. In the multiplexed setting, the stronger end-to-end conditions become \(w_{e2e}\geq w\) and \(R_{e2e}^{(m)} \geq p^{\mathrm{yield}}\), where \(p^{\mathrm{yield}} = m\,p\) denotes the target yield inherited from the elementary-link layer, since each logical connection has access to \(m\) parallel edge-disjoint entanglement generation attempts. Correspondingly, we define the end-to-end yield rate as

\begin{equation}
\begin{aligned}
R_{\mathrm{e2e}}^{(m)}
&= \mathbf{1}_{\mathcal{F}}
\sum_{n=r}^{mD}
\Pr[N=n]\,
\left\lfloor \frac{n}{r} \right\rfloor
p_{\mathrm{suc}},\\
\mathcal{F}
&=\left\{
w_{\mathrm{e2e}}^{(r)}\!\left(w^{D}\right)\geq w
\right\},\\
N&\sim \mathrm{Binomial}\!\left(mD,p^{D}\right),
\qquad D=k+1 .
\end{aligned}
\end{equation}
where \(N\) is the number of available input copies for purification, \(r\) is the purification block size, and \(p_{\mathrm{suc}}\) is the success probability of one \(r\)-to-\(1\) purification attempt; Eq.~\eqref{eq:4-to-1} gives this probability for the \(4\)-to-\(1\) case as an example. This definition explicitly accounts for the fact that the elementary links are edge-disjoint entanglement resources, and therefore the corresponding end-to-end requirement must be expressed in terms of end-to-end yield rather than only single-pair success.

The required link-level resources for the multiplexed case are summarized in Table~\ref{tab:bcube_yield_minm_maxm_ieee}. Since the rate constraint also depends on the Werner parameter, for each network setting we first identify the minimum non-trivial link-level Werner parameter for which the end-to-end constraint is satisfied. Then, at that minimum feasible answer, we search for the minimum link-level entanglement generation rate that satisfies both the Werner and the rate constraints.

Taken together, Tables~\ref{tab:bcube_anyw_nontrivialp_k10} and~\ref{tab:bcube_yield_minm_maxm_ieee} show that the resource requirements of unbounded-scale hop-independent fidelity in B-Cube is governed by a fundamental tradeoff between topology-enabled redundancy and purification difficulty over longer paths. As \(k\) increases, a diameter-apart pair gains access to up to \(k+1\) edge-disjoint shortest paths, which improves the ability to generate enough number of copies for multiple parallel end-to-end purification attempt. At the same time, increasing \(k\) also increases the end-to-end path length, which degrades the quality of the raw states entering purification and raises the purification threshold \(w_{\mathrm{th}}(k)\). This tradeoff explains the different trends observed in the minimum required link-level Werner parameter.

In the non-multiplexed case \((m=1)\), the system remains rate-limited throughout the scanned range, as indicated by the fact that the reported \(\min w\) values stay well above \(w_{\mathrm{th}}(k)\). Consequently, the dominant effect of increasing \(k\) is the growth in path redundancy, and \(\min w\) decreases monotonically. In the multiplexed case, however, the additional parallelism relaxes the rate bottleneck much earlier. As a result, \(\min w\) initially decreases with \(k\), but beyond a crossover point the rate constraint ceases to dominate and the active lower bound becomes \(w_{\mathrm{th}}(k)\), which increases with \(k\). This explains why the resource requirement in multiplexed cases first decreases and then increases, whereas the non-multiplexed curve only decreases.

The behavior of the minimum required link-level success probability \(\min p\) is governed by the same competition, but with end-to-end yield as the relevant rate quantity. Increasing \(k\) improves the copy-availability term by providing more candidate shortest paths, but it worsens the purification-success term because the raw end-to-end states are distributed over longer paths and, in the feasible region, are associated with smaller Werner parameters. For the present \(r=4\) purification setting, the key structural quantity is \(\lfloor (k+1)/4 \rfloor\), which gives the number of full \(4\)-to-\(1\) purification groups that can be formed. In particular, \(k=7\) is the first point at which two full purification groups become available, which explains the continued decrease in \(\min p\) up to that point. By contrast, increasing from \(k=7\) to \(k=8\) does not create an additional full purification group; it only introduces one spare path while the purification inputs become worse, so \(\min p\) increases. For \(k=9\) and \(k=10\), the growing number of spare paths improves the probability of successfully supplying both purification groups despite probabilistic copy generation, causing \(\min p\) to decrease again. Thus, the observed trends in both \(\min w\) and \(\min p\) reflect the same underlying balance between increasing topological redundancy and the growing challenge of purifying lower-quality long-distance end-to-end states.

Ultimately, what gives these results their weight is not the size of the required resources, but what those resources enable. Hop-independent fidelity decouples the scale of a quantum data center from the quality of its individual QPUs. Because the isometric B-Cube ((n,k)) construction keeps the number of swaps on any connection bounded and provides end-to-end redundancy that grows with distance exactly where it is needed, the link-level resource quality that the network must sustain does not increase as the network grows. From the application’s perspective, the number of swaps between any two QPUs remains bounded, regardless of where those QPUs are physically located. This means that computation can be placed anywhere in the data center without incurring a distance-dependent fidelity penalty. The network can then be expanded by adding more nodes of the same quality, rather than by requiring better nodes. This reflects the server-centric principle of classical data center networking: scaling out with commodity components rather than scaling up with increasingly powerful ones, now carried into the much harder setting of entanglement distribution.

The benchmark is admittedly high. The required fidelity is demanding but still within plausible reach of current sources, while the per-attempt success probability, near (0.98) even at (m=5), is the harder requirement and, under deterministic swapping with probabilistic generation, the binding constraint. The decisive point, however, is that this requirement is fixed. It is a threshold that must be cleared once, not a slope that becomes steeper each time more nodes are added. Clearing this threshold turns the open-ended problem of preserving entanglement quality over growing distance into a single, scale-independent requirement on the elementary link, allowing the quantum data center to grow without bound and without continual hardware upgrades.

\section{Discussion}
\label{sec:discussion}

In this work, we focused on the BCube topology because its path redundancy grows with Hamming distance: the connections that experience the largest raw fidelity degradation also provide the largest number of edge-disjoint shortest paths. This structure is well matched to purification-based recovery of link-level Werner quality. We therefore first determine the minimum value of $r$ for which a single $r$-to-1 purification block restores the Werner parameter to the elementary-link target. When additional raw copies are available, they are used to run multiple such $r$-to-1 purification blocks in parallel, rather than being distributed across deeper recursive purification layers, since purification becomes less likely to succeed as the input Werner parameter decreases. These choices, however, leave open broader design questions. Other topology families, adaptive purification schedules, alternative resource-allocation rules, or higher-order multi-copy protocols may achieve the same hop-independent fidelity target with lower overhead. More generally, our results suggest that feasibility is shaped jointly by topology, purification, and available link-level resources. Identifying designs that realize this feasible region with minimal overhead is therefore an important direction for future work.

Beyond these design choices, the present results rest on several modeling assumptions, and relaxing each one defines a natural extension of the analysis. The first concerns the type of resource considered. Our analysis treats the network primarily as a system of elementary bipartite entanglement resources, whereas a more general quantum network may have both bipartite and multipartite entanglement available. Moreover, resources that are not selected during routing and swapping are discarded in our model. A more complete architecture could instead use quantum memories to store and reuse unused entangled states for later routing, purification, or higher-level connection requests.

A second assumption concerns device behavior. Our analysis assumes probabilistic elementary entanglement generation but deterministic entanglement swapping. This assumption is useful because it isolates the role of topology and path redundancy in enabling hop-independent Werner-parameter recovery. However, realistic quantum networks will involve imperfect swapping, heterogeneous elementary-link qualities, and memory decoherence. Extending the model to include these effects would test whether the proposed resource-scaling behavior remains robust under device-level constraints. In particular, it would show whether stronger topological redundancy can compensate not only for hop-induced Werner-parameter degradation, but also for additional probabilistic losses from realistic network operations.

Finally, a third assumption concerns the order of operations: our protocol follows a swap-then-purify order at the module level. This choice does not exclude link-level purification, since any purification performed before elementary-link generation can be absorbed into the elementary-link Werner parameter $w$ used as the input to our model. We focus on swap-then-purify because it supports a shorter time-step structure: raw end-to-end copies are first generated across the module, and purification is then applied only at the end points. In contrast, a purify-then-swap strategy, or an alternating sequence of purification and swapping, would introduce additional synchronization rounds and increase the number of time steps required to establish a module-level resource. Comparing these timing and resource tradeoffs is an important direction for future work.

\section{Conclusion}
\label{sec:conclusion}

This paper studied whether the scalability advantages of server-centric quantum data-center networks can be made compatible with fidelity preservation. In such architectures, increasing the network size typically increases the number of hops and Bell-state measurements between distant QPUs, which in turn lowers the raw end-to-end fidelity. We focused on a stronger target than simply meeting an application-specific fidelity threshold: hop-independent fidelity, where a multi-hop end-to-end state is restored to the quality of a successfully generated elementary link. Achieving this target allows the network to hide its internal path length and expose each QPU pair as an effective link of uniform quality.

We approached this question in three steps. First, we used a topology-independent black-box path model to determine how many raw end-to-end copies are needed to recover the elementary-link Werner parameter over a distance-$\ell$ connection. This analysis showed that optimized nested (r)-to-1 purification protocols can substantially reduce the required operating region compared with recursive 2-to-1 purification. Second, we used these requirements to motivate the HI-QDC architecture: a recursive, isometric, modular construction in which each ($\mathrm{BCube}(n,k)$) module converts its internal multi-hop entanglement into an effective link-level resource for an outer module of the same structural form. Because each recursive level has the same topology and the same statistical resource description, the purification condition verified for one module can be reused at higher levels under the same operating assumptions.

The topology-aware analysis then identified the elementary-link resources required for a BCube module to act as such an effective link. Recovering the Werner parameter alone becomes feasible without multiplexing starting at ($\mathrm{BCube}(n,3)$), while preserving both elementary-link fidelity and elementary-link yield is first achieved at (k=4). Across the feasible network sizes studied up to (k=10), this stronger condition is met using a fixed purification size of (r=4). The resulting resource picture is not one in which fidelity, success probability, and multiplexing can freely substitute for one another. The elementary-link Werner parameter sets the feasibility floor, multiplexing provides the most direct way to increase the number of available copies while it remains modest, and the elementary-link success probability controls whether the purified end-to-end yield can match the link-level benchmark. Most importantly, these requirements appear as fixed thresholds to be satisfied by the elementary links, not as demands that grow with every additional QPU.

This fixed-threshold behavior is the main implication of hop-independent fidelity. Once the elementary links operate inside the feasible region, the quality required from individual QPUs and elementary connections does not need to improve simply because the data center grows. The isometric BCube-based construction keeps the relevant hop count bounded within each module, while its path diversity supplies additional end-to-end copies exactly for the longer connections that suffer greater raw fidelity loss. As a result, computation can be placed across the data center without introducing a distance-dependent fidelity penalty, and the network can grow by adding more QPUs of the same quality rather than by requiring increasingly better ones. This brings the server-centric scaling principle of classical data centers into the quantum setting, where the central challenge is not only connectivity but also preservation of entanglement quality.

These conclusions hold under the assumptions used throughout the analysis: probabilistic elementary-link generation, deterministic swapping, homogeneous elementary-link resources, and purification applied after end-to-end copies have been generated. Relaxing these assumptions is a natural next step. Imperfect swapping, heterogeneous link qualities, memory decoherence, and alternative purification schedules will determine how robust the feasible region remains under more detailed device-level constraints.

The central lesson is that the same growth that makes server-centric networks vulnerable to hop-induced fidelity loss can also provide the redundancy needed to repair that loss. In this sense, topology supplies the parallel end-to-end resources required by purification, and purification turns those resources into hop-independent entanglement quality.

\section*{Acknowledgment}

KPS thanks the U.S. Department of Energy, Office of Science, Advanced Scientific Computing Research (ASCR) program, for support under Award Number DE-SC0026264, and PQI Community Collaboration Awards. The authors also thank Inder Monga, Reza Nejabati and Rodney Van Meter for helpful discussions.

\bibliographystyle{IEEEtran}
\bibliography{qce_intro_core}

\appendix
\label{sec:appendix}
This appendix collects the formal results referenced in Section~\ref{sec:Architecture}: the equality between the number of edge-disjoint shortest paths and the path length, the consequent homogeneity of fidelity profiles, and the station-count comparison underlying the resource benefit. Throughout, a $\mathrm{BCube}(n,k)$ addresses its QPUs by $(k{+}1)$-tuples $x=(x_1,\dots,x_{k+1})$ over $\{0,\dots,n{-}1\}$. Two QPUs share their layer-$i$ switch if and only if their addresses agree in every coordinate except the $i$-th, so a single hop through a layer-$i$ switch changes only coordinate $i$. For QPUs $u$ and $v$, we write $\Pi=\{\,i:u_i\neq v_i\,\}$ for the set of differing coordinates and $d=|\Pi|$ for their Hamming distance.

We first record that the shortest path realizes this Hamming distance.

\begin{lemma}
\label{lem:length}
The shortest path between $u$ and $v$ has length $d$, and each shortest path corresponds to an ordering of the $d$ coordinates in $\Pi$.
\end{lemma}

\begin{IEEEproof}
Each hop alters exactly one coordinate, so reaching $v$ from $u$ requires correcting each of the $d$ coordinates in $\Pi$ at least once. Hence every path has length at least $d$. Correcting each differing coordinate exactly once yields a path of length exactly $d$, so the shortest-path length is $d$. Such a path is determined by the order in which the coordinates of $\Pi$ are corrected.
\end{IEEEproof}

\begin{theorem}
\label{thm:disjoint}
In $\mathrm{BCube}(n,k)$, let $u$ and $v$ be two QPUs at Hamming distance
$d$, and let
\[
\Pi=\{\,i:u_i\neq v_i\,\}.
\]
The maximum number of pairwise edge-disjoint shortest paths between $u$ and
$v$ is exactly $d$.
\end{theorem}

\begin{IEEEproof}
By Lemma~\ref{lem:length}, every shortest path from $u$ to $v$ has length
$d$ and corrects exactly the $d$ coordinates in $\Pi$, each exactly once.
Therefore, the first hop of any shortest $u$--$v$ path must correct one
coordinate $i\in\Pi$.

In $\mathrm{BCube}(n,k)$, each QPU has exactly one physical edge to its
layer-$i$ switch. Hence two shortest paths whose first corrected coordinate
is the same $i\in\Pi$ must both use the same physical edge incident to $u$.
Such paths cannot be edge-disjoint. Since there are only $d$ choices for the
first corrected coordinate, any family of pairwise edge-disjoint shortest
paths has size at most $d$.

It remains to show that this upper bound is attainable. Arrange the differing
coordinates in a fixed cyclic order
\[
\pi_0,\pi_1,\dots,\pi_{d-1}.
\]
For each $j\in\{0,\dots,d-1\}$, define $P_j$ to be the shortest path that
corrects the coordinates in the cyclically rotated order
\[
\pi_j,\pi_{(j+1)\bmod d},\dots,\pi_{(j+d-1)\bmod d}.
\]
Each $P_j$ has length $d$ by Lemma~\ref{lem:length}.

We now show that the paths $P_0,\dots,P_{d-1}$ are edge-disjoint. Let
$x_{j,t}$ denote the QPU reached after the first $t$ hops of $P_j$, where
$1\le t\le d-1$. Then $x_{j,t}$ agrees with $v$ on the corrected set
\[
S_{j,t}
=
\{\pi_j,\pi_{(j+1)\bmod d},\dots,\pi_{(j+t-1)\bmod d}\}
\]
and agrees with $u$ on all remaining coordinates. Since $u_i\neq v_i$ for
all $i\in\Pi$, the corrected set $S_{j,t}$ uniquely determines the
intermediate QPU $x_{j,t}$.

Suppose two paths $P_j$ and $P_{j'}$ meet at an intermediate QPU. Then for
some $1\le t,t'\le d-1$ we have
\[
x_{j,t}=x_{j',t'}.
\]
Therefore,
\[
S_{j,t}=S_{j',t'}.
\]
Both sets are nonempty proper contiguous arcs in the cyclic order
$\pi_0,\dots,\pi_{d-1}$. Equality of these arcs implies first that
$t=t'$, since the two sets have the same cardinality. Because the arcs are
proper, the starting point of each arc is the unique element whose cyclic
predecessor is not contained in the arc. Hence the starting indices must also
agree, so $j=j'$. Therefore, no two distinct paths share an intermediate QPU.

Finally, the first edges of the paths are distinct because $P_j$ leaves $u$
through the layer-$\pi_j$ switch, and the coordinates $\pi_j$ are distinct as
$j$ varies. Similarly, the last edges are distinct because $P_j$ enters $v$
through the layer-$\pi_{(j-1)\bmod d}$ switch, which is also distinct as
$j$ varies. Since the paths share no intermediate QPU and use distinct edges
incident to $u$ and $v$, they are pairwise edge-disjoint.

Thus there exist $d$ pairwise edge-disjoint shortest paths, and no family of
pairwise edge-disjoint shortest paths can contain more than $d$ paths. The
maximum number is therefore exactly $d$.
\end{IEEEproof}

\begin{remark}
The word ``shortest'' is essential in Theorem~\ref{thm:disjoint}. For
non-shortest paths, a path may first change a coordinate on which $u$ and
$v$ already agree and later change it back. Therefore, when $d<k+1$, the
maximum number of edge-disjoint $u$--$v$ paths of arbitrary length can exceed
$d$.
\end{remark}

The cyclic structure of these paths also fixes their fidelity. Let $\Phi(P)=\prod_{e\in P}\phi(e)$ denote the fidelity of a path, where $\phi(e)$ is the fidelity of edge $e$.

\begin{theorem}
\label{thm:fidelity}
If every link produced by a station in layer $i$ has the same fidelity $\phi_i$, then each of the edge-disjoint shortest paths $P_0,\dots,P_{d-1}$ of Theorem~\ref{thm:disjoint} satisfies
\[
\Phi(P_j)=\prod_{i\in\Pi}\phi_i,
\]
independently of $j$. In particular, all $d$ disjoint paths share the same end-to-end fidelity, even when the per-layer fidelities $\{\phi_i\}_{i\in\Pi}$ differ.
\end{theorem}

\begin{IEEEproof}
By construction, each $P_j$ corrects every coordinate of $\Pi$ exactly once. Hence it traverses exactly one link at each layer $i\in\Pi$ and none elsewhere. Under the layer-uniform assumption, that link has fidelity $\phi_i$ regardless of which station originates it. Therefore,
\[
\Phi(P_j)=\prod_{i\in\Pi}\phi_i,
\]
where the product is taken over the layers in $\Pi$ visited in some order. Because multiplication is commutative, the product depends only on the multiset $\{\phi_i\}_{i\in\Pi}$ and not on the order, so it is identical for every $j$.
\end{IEEEproof}

It remains to count network resources, taking BSM stations as the resource and treating the multiplexers as cost-free. A single $\mathrm{BCube}(n,k)$ has $k{+}1$ layers, each containing $n^{k}$ stations, one per choice of the $k$ address coordinates other than the layer's own. It therefore contains $(k{+}1)n^{k}$ stations while hosting $n^{k+1}$ QPUs.

Hosting $N=n^{2k+1}=n^{(2k)+1}$ QPUs by classical scaling requires a single $\mathrm{BCube}(n,2k)$, with
\begin{equation}
\label{eq:flat}
  S_{\mathrm{flat}} = (2k+1)n^{2k}
\end{equation}
stations. HI-QDC instead uses $n^{k}$ inner $\mathrm{BCube}(n,k)$ modules, contributing
\[
n^{k}(k{+}1)n^{k}=(k{+}1)n^{2k}
\]
stations, together with one outer $\mathrm{BCube}(n,k)$ over the $n^{k+1}$ multiplexed ports. Since the first layer of the outer structure is supplied by the inner connectivity, only $k$ of its $k{+}1$ layers are realized, contributing $k n^{k}$ stations. Hence
\begin{equation}
\label{eq:hi}
  S_{\mathrm{HI}} = (k+1)n^{2k} + k n^{k}.
\end{equation}
Dividing \eqref{eq:hi} by \eqref{eq:flat} gives
\begin{equation}
\label{eq:ratio}
  \frac{S_{\mathrm{HI}}}{S_{\mathrm{flat}}}
  =
  \frac{k+1}{2k+1}
  +
  \frac{k}{(2k+1)n^{k}}
  \xrightarrow[k\to\infty]{}
  \frac12 .
\end{equation}
For $n=2$ and $k=2$, this gives $S_{\mathrm{flat}}=5\cdot2^{4}=80$ and $S_{\mathrm{HI}}=3\cdot2^{4}+2\cdot2^{2}=56$, corresponding to a $30\%$ reduction, alongside $n^{k}=4$ cost-free multiplexers.

\end{document}